\documentclass{vldb}
\usepackage{graphicx}
\usepackage{balance}  % for  \balance command ON LAST PAGE  (only there!)
\usepackage{color}
\usepackage{multirow}
\usepackage{subfigure}
\usepackage{times}

\newcommand{\stitle}[1]{\vspace{0.8ex}\noindent\textup{\textbf{#1}}}

\usepackage{algorithm}
\usepackage{algpseudocode}
\newtheorem{lemma}{Lemma}

 \newtheorem{definition}{Definition}

\newcommand{\hide}[1]{}

\def\bitcoinA{%
  \leavevmode
  \vtop{\offinterlineskip %\bfseries
    \setbox0=\hbox{B}%
    \setbox2=\hbox to\wd0{\hfil\hskip-.03em
    \vrule height .3ex width .15ex\hskip .08em
    \vrule height .3ex width .15ex\hfil}
  \vbox{\copy2\box0}\box2}}

\def\bitcoinB{\leavevmode
  {\setbox0=\hbox{\textsf{B}}%
    \dimen0\ht0 \advance\dimen0 0.2ex
    \ooalign{\hfil \box0\hfil\cr
      \hfil\vrule height \dimen0 depth.2ex\hfil\cr
    }%
  }%
}

\pdfoutput=1

\makeatletter
\def\@copyrightspace{\relax}
\makeatother

\begin{document}

% ****************** TITLE ****************************************

\title{Flow Computation in Temporal Interaction Networks}

% possible, but not really needed or used for PVLDB:
%\subtitle{[Extended Abstract]
%\titlenote{A full version of this paper is available as\textit{Author's Guide to Preparing ACM SIG Proceedings Using \LaTeX$2_\epsilon$\ and BibTeX} at \texttt{www.acm.org/eaddress.htm}}}

% ****************** AUTHORS **************************************


\numberofauthors{1} %  in this sample file, there are a *total*

\author{
% 1st. author
\alignauthor
Chrysanthi Kosyfaki, Nikos Mamoulis, Evaggelia Pitoura, Panayiotis Tsaparas\\
       \affaddr{Department of Computer Science and Engineering}\\
       \affaddr{University of Ioannina}\\
       %\affaddr{Ioannina, Greece}\\
       \email{\{xkosifaki,nikos,pitoura,tsap\}@cse.uoi.gr}
       }
% 2nd. author

% 5th. author
%\alignauthor Sean Fogarty\\
      % \affaddr{NASA Ames Research Center}\\
       %\affaddr{Moffett Field}\\
       %\affaddr{California 94035}\\
       %\email{fogartys@amesres.org}
% 6th. author
%\alignauthor Charles Palmer\\
      % \affaddr{Palmer Research Laboratories}\\
       %\affaddr{8600 Datapoint Drive}\\
       %\affaddr{San Antonio, Texas 78229}\\

% There's nothing stopping you putting the seventh, eighth, etc.
% author on the opening page (as the 'third row') but we ask,
% for aesthetic reasons that you place these 'additional authors'
% in the \additional authors block, viz.
%\additionalauthors{Additional authors: John Smith (The Th{\o}rv\"{a}ld Group, {\texttt{jsmith@affiliation.org}}), Julius P.~Kumquat
%(The \raggedright{Kumquat} Consortium, {\small \texttt{jpkumquat@consortium.net}}), and Ahmet Sacan (Drexel University, {\small \texttt{ahmetdevel@gmail.com}})}
% Just remember to make sure that the TOTAL number of authors
% is the number that will appear on the first page PLUS the
% number that will appear in the \additionalauthors section.

\maketitle

\begin{abstract}
Temporal interaction networks capture the history of activities between
entities along a timeline.
At each interaction, some quantity of data
(money, information, kbytes, etc.) {\em flows} from one vertex of the
network to another.
Flow-based analysis
%of interaction networks
can reveal
important information.
%(e.g., cyclic transactions).
For instance,
%transactions networks such as
financial intelligent units (FIUs) are
interested in finding subgraphs in transactions networks with
significant flow of money transfers.
%Computing
%the flow in such networks can facilitate their analysis.
%t for studying in depth large networks.
In this paper, we introduce
%and
%efficiently solve
the flow computation problem
in an interaction network or a subgraph thereof.
We propose and study two models of flow computation, one based on a
greedy flow transfer assumption and one that finds the maximum possible
flow.
We show that the greedy flow computation problem can be easily solved
by a single scan of the interactions in time order.
For the harder maximum flow problem, we propose graph
precomputation and simplification approaches that can greatly reduce its
complexity in practice.
%based on a greedy transfer assumption.
%We also discuss and analyze the problem of finding the maximum
%possible flow transfer throughout such a network, if the interactions do
%not necessarily transfer the maximum possible quantity. We model this
%as a linear programming (LP) problem and identify its equivalence to
%a maximum flow computation problem in temporal networks.
%We identify the classes of directed acyclic graphs for which greedy flow computation
%solves the maximum flow transfer problem.
%In addition, we propose an efficient graph preprocessing
%algorithm, which removes all interactions, edges and vertices
%that do not contribute to the maximum flow
%computation.
%Finally, we propose a graph simplification approach, which uses the
%efficient greedy algorithm to compute part of the maximum flow, before
%applying LP to solve the remainder of the problem.
As an application of flow computation,
we formulate and solve the problem of {\em flow pattern} search,
where, given a graph pattern,  the objective is  to find
its instances and their flows in a large
interaction network.
%and then compute the flow for each pattern instance. 
%as directed acyclic graphs and study their enumeration in large
%interaction networks. We propose an efficient pattern enumeration
%algorithm and compare it with baseline alternatives.
We evaluate our algorithms using real datasets. The results show that
the techniques proposed in this paper can
greatly reduce the cost of flow computation and pattern
enumeration.
\end{abstract}

\section{Introduction}
\textit{Temporal interaction networks}
model the transfer of data quantities between entities along a timeline.
At each interaction, a quantity (money,
messages, kbytes etc.) {\em flows} from one network vertex (entity) to another.
Analyzing interaction networks can reveal important information (e.g.,
cyclic transactions, message interception).
For instance, financial intelligent units (FIUs) are often interested
in finding subgraphs of a transaction network, wherein vertices
(financial entities)
have exchanged a significant amount of money
directly or through intermediaries
(e.g., multiple small-volume transactions through a subnetwork
that aggregate to large
amounts).
Such exchanges may be linked to criminal behavior (e.g., money laundering).
In addition, significant money flow in the bitcoin transaction network has been associated to 
theft \cite{DBLP:conf/imc/MeiklejohnPJLMVS13}.

%along a time interval.

\stitle{Problem.}
In this paper, we study the problem of computing the flow through
an interaction network (or a sub-network thereof),
from a designated vertex $s$, called {\em source} to a designated
vertex $t$, called {\em sink}.
%Consider, for instance, the network shown in
As an example, Figure \ref{fig:introex}(a) shows a toy interaction network,
where vertices are bank accounts and each edge is a sequence of
interactions in the form $(t_i,q_i)$, where $t_i$ is a timestamp and
$q_i$ is the transferred quantity (money).
%Assume that the objective is to compute the amount of money
%transferred from vertex $s$ (source) to vertex $t$ (sink) throughout
%the network.
To model and solve the flow computation problem from $s$ to $t$,
we assume that
throughout the history of interactions, each vertex $v$ has a {\em
  buffer} $B_v$.
Since we are interested in measuring the flow from $s$ to $t$, we
assume that initially $s$ has infinite buffer and that
the buffers of all other vertices are $0$. 
Interactions are examined in order of time and,
as a result of an interaction $(t_i,q_i)$ on edge $(v,u)$,
vertex $v$ may transfer from $B_v$ to $u$'s
buffer $B_u$ a quantity in $[0,q_i]$. 
For example, if interaction $(1,\$3)$ on edge $(s,x)$ transfers $\$3$
from $B_s$ to $B_x$, interaction $(5,\$5)$ on edge  $(x,z)$
can transfer at most $\$3$ from $B_x$ to $B_z$.
After the end of the timeline, the buffered quantity at the
sink vertex $t$ models the flow that has been transferred from $s$ to
$t$.

\begin{figure}[tbh]
  \centering 
  \small
  \begin{tabular}{@{}c c@{}}
  \includegraphics[width=0.46\columnwidth]{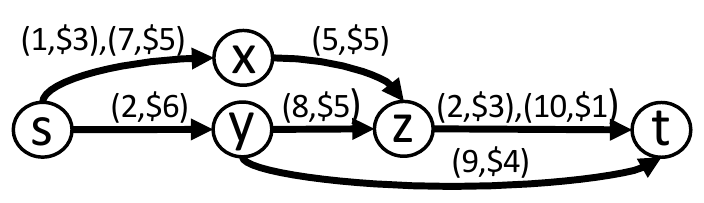}&
  \includegraphics[width=0.46\columnwidth]{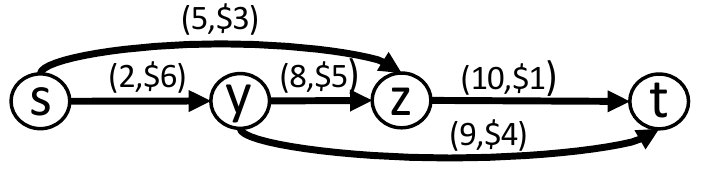}\\
    (a) interaction network&
         (b) simplified network\\
  \end{tabular}
  \caption{A toy interaction network}
  \label{fig:introex}
\end{figure}

We study two models of flow transfer as an effect of an interaction
$(t_i,q_i)$ on an edge $(v,u)$.
The first one is based on a {\em greedy flow transfer}
assumption, where $v$ transfers to $u$ the
maximum possible quantity, i.e., $\min\{q_i,B_v\}$.
According to the second model, $v$ may transfer to $u$
any quantity in $[0, \min\{q_i,B_v\}]$, {\em reserving} the remaining
quantity for future outgoing interactions from $v$ (to any 
vertex). The objective is then to compute the {\em maximum flow} that
can be transferred from $s$ to $t$.
In our running example, as a result of interaction $(2,\$6)$ on edge
$(s,y)$, $\$6$ are transferred from $B_s$ to $B_y$.
As a result of interaction $(8,\$5)$ on edge
$(y,z)$, the greedy model transfers 
$\$5$ (i.e., the maximum possible amount) from
$B_y$ to $B_z$, leaving only $\$1$ to $B_y$ for future interactions.
Hence, interaction $(9,\$4)$ on edge $(y,t)$ may only transfer $\$1$
to the buffer $B_t$ of the sink.
Had interaction $(8,\$5)$ transferred only $\$2$ from $B_y$ to $B_z$,
interaction $(9,\$4)$ would be able to transfer $\$4$ from 
$B_y$ to $B_t$. This decision would maximize the flow that reaches $t$.
%\todo{add example and show that greedy flow cannot always compute max flow}

\stitle{Applications.}
Flow computation in interaction networks finds application in
different domains. As already discussed, computing the flow
of money from one financial entity (e.g., back customer,
cryptocurrency user) to another can help in defining their
relationship and the roles of any intermediaries in them \cite{kondor2014inferring}. 
As another application, consider a transportation network (e.g.,
flights network, road network) and the problem of computing the
maximum flow (e.g., of vehicles or passengers) from a source to a
destination. Identifying cases of heavy flow transfer can help in
improving the scheduling or redesigning the network.
Similarly, in a communications network, 
measuring the flow between vertices (e.g., routers) can help in
identifying abnormalities (e.g., attacks) or bad design.
Recent studies in cognitive science \cite{Chen2019}
associate the information flow in
the human brain with the embedded network topology and the
interactions between different (possibly distant) regions.
%In social networks \todo{+++}

\stitle{Contributions.}
We show that the greedy flow computation problem can be solved very fast by
performing only a linear scan of all interactions in order of time
and updating two buffers at each interaction.
However, greedy flow computation is of less interest and has fewer
applications
compared to the
more challenging maximum flow
computation problem. The latter is not always solved by greedy flow
computation, as we have shown
with the example of Figure \ref{fig:introex}(a).
We show how the maximum flow problem
can be formulated and solved using linear programming (LP).
In a nutshell, we
can define one variable for each interaction (except from those originating
from the source vertex $s$, which always transfer the maximum possible
flow) and find the values of the variables that
maximize the total flow that ends up at the sink.

Flow computation in networks is not a new problem, however, previous
work has mainly focused on
the classic maximum flow problem in a
static graph, where vertices are junctions and edges have capacities
\cite{cormen2009introduction}.
Our problem setup is quite different, since
our vertices model entities and
edges are time-series of interactions, each of which happens at a specific
timestamp.
However, as we show, our problem turns out to be
equivalent to a temporal flow computation
problem, where the edges
(and their capacities) are ephemeral.
%This problem can be solved by linear programming
Akrida et al. \cite{DBLP:conf/ciac/AkridaCGKS17}
show that this temporal flow computation problem is equivalent to flow
computation in a static graph, where an edge is defined for each 
ephemeral edge, meaning that the complexity of our problem
is quadratic to the number of interactions (e.g., if Edmonds–Karp
algorithm \cite{EdmondsK72} is used).
%As shown in \cite{DBLP:conf/ciac/AkridaCGKS17}, the temporal flow computation
%problem can be reduced to a static flow computation problem, where
%versions of the vertices and the edges are defined based on
%the time intervals that the ephemeral edges are active in the original
%network.
%The number of edges in the reduced static graph is the same
%by order of time
%(for example the complexity of the well-known algorithm
%\textit{Ford-Fulkerson} is $O(Emax|f^*|)$ which $f^*$ is the
%maximum flow found by the algorithm \cite{cormen2009introduction}).

Hence, solving the maximum flow problem on a  network with numerous
interactions can be quite expensive.
%, because the
%(superlinear)
%its
%complexity depends on the number of interactions \todo{++}.
We propose a set of techniques that reduce the cost in practice.
First, we show that for certain classes of networks (such as simple paths), the
greedy algorithm can compute (exactly)
the maximum flow in linear time to the number of interactions.
Verifying whether greedy can compute the maximum flow costs only
a single pass over the vertices.  
Second, we propose a preprocessing algorithm that
eliminates interactions, edges and vertices that cannot contribute to
the maximum flow, with a potential to greatly reduce the problem size
and complexity. For example, interaction $(2,\$3)$ on edge $(z,t)$ of
the network in Figure \ref{fig:introex}(a) can be eliminated because
all incoming interactions to $z$ have timestamps greater than $2$;
hence, interaction $(2,\$3)$ cannot transfer any incoming quantity to $z$.
Third, we design an algorithm that performs greedy flow
computation on a part of the graph, simplifying the graph on which LP
has to be eventually applied. For example, the path formed by
edges $(s,x)$ and $(x,z)$
can be reduced to a single edge $(s,z)$ as shown in Figure
\ref{fig:introex}(b), because not propagating the maximum possible
flow through this path to $z$ and reserving flow at $s$ or $x$ cannot
increase the maximum flow that eventually reaches $t$. 
We conduct an experimental evaluation, where we compute the flow
on subgraphs of three large real networks and show that our maximum
flow computation approach is very effective, achieving at least one order
of magnitude cost reduction compared to the baseline LP algorithm.

As an application of flow computation, we also formulate and
study the problem of flow pattern search in large interaction
networks. These patterns as small graphs that repeat themselves
in the network. The problem is to find the instances and compute
the flow for each of them. We propose a 
graph preprocessing approach that facilitates the enumeration of certain
classes of patterns and their maximum flows.
As we show experimentally, enumerating the instances of flow
patterns and computing their flow can greatly benefit from the
precomputed data.
%\todo{mention applications}

The contributions of this paper can be summarized as follows:
\begin{itemize}
  \item This is the first work, to our knowledge, which studies flow
    computation in temporal interaction networks. We propose two models for
    flow computation and analyze their complexities. The first
    model comes together with a greedy computation algorithm, while
    maximum flow computation can be formulated and solved as a linear
    programming problem.
  \item For maximum flow computation, which can be
    expensive, we propose (i) an efficient check for verifying if
    it can be solved exactly by the greedy algorithm, (ii) a   
    graph preprocessing technique, which can eliminate interactions,
    vertices and edges from the graph, (iii) a graph simplification
    approach, which reduces part of the graph to edges, the flow of
    which can be derived using the greedy algorithm.
  \item We formulate and study a flow pattern enumeration problem
    which computes instances of graph patterns in a large graph and
    their flows. We show that a graph preprocessing technique can
    accelerate pattern enumeration.
  \item We conduct experiments using data from three real
    interaction networks to evaluate our techniques.
  \end{itemize}

  \stitle{Outline.}
  The rest of the paper is organized as follows. Section
\ref{sec:related} reviews work related to flow computation on static
and temporal graphs and to pattern enumeration on large networks.
Section \ref{sec:def} defines basic concepts and
Section \ref{sec:flowcomp} defines flow computation models and
algorithms.
In Section \ref{sec:algorithm}, we study the problem of flow pattern
search.
A thorough experimental evaluation on real data is presented in
Section \ref{sec:experiments}.
Finally, Section \ref{sec:conclusion} concludes the paper with
directions for future work.

%\todo{add example}

\section{Related Work}\label{sec:related}
There have been numerous studies that investigate the problems of flow computation in networks and enumeration of patterns in graphs. In this section, we summarize the most representative works for the above problems and discuss their relation to our study.
\subsection{Flow Computation}
Ford and Fulkerson \cite{Ford:2010:FN:1942094} were the first who tracked the max-flow problem.
Given a DAG, with a {\em source} node $s$ with no incoming edges and a {\em sink} node $t$ with no outgoing edges and assuming that each edge has a {\em capacity},
the \textit{Ford-Fulkerson algorithm}
finds the maximum flow that can be transferred from $s$ to $t$ through the edges of the network, assuming that each edge has a maximum {\em capacity} for flow transfer.
% detects the maximum flow from a source vertex to a destination vertex.
The algorithm applies on {\em static} networks, in which the existence of edges and their capacities do not change over time. In addition, the flow is assumed to be transferred instantly from one vertex to another and to be constant over time.
Since then, a number of models and algorithms for maximum flow computation have been developed  \cite{ahuja1995network,goldberg2014efficient}.
%(have timestamps on the edges
%(no changes are y do not consider about the time that the flow carries from the source to destination but for the maximum amount of it).

Skutella \cite{Skutella2008AnIT} surveyed temporal maximum flow computation
problems.
%, by outlining problem definitions and solutions.
In these problems, each edge, besides having a capacity, is characterized by a {\em transit time}, i.e., the time needed to transfer flow equal to its capacity \cite{hoppe1995efficient}.
%are labeled by the time when they connect with vertices.
The general problem is to find the maximum flow that can be transferred from $s$ to $t$ within a time horizon $T$
%, by considering the time needed for transiting
\cite{baumann2009earliest,DBLP:journals/networks/RuzikaSS11}.
%\cite{goldberg2014efficient}
%introduced the concept of max-flow problem in temporal networks. In his work, he considers not only the flow that transfers but also the time that this happens.
In another model for temporal flow computation, each edge is assumed to be {\em ephemeral}, i.e., it cannot be used to transfer flow at any time. Akrida et al. \cite{DBLP:conf/ciac/AkridaCGKS17} studied the max-flow problem in such networks. They assume that each edge is valid at certain {\em days} (e.g., day 5 and day 8).
Similar to  \cite{Skutella2008AnIT}, the problem is to find the maximum flow that can be transferred from $s$ to $t$ by the end of day $l$. The capacity of an edge is the amount of flow that it can transfer each day that it is valid. The vertices of the network have a buffer, meaning that they can hold a maximum amount of flow before this can be transferred by an outgoing edge that will become available in the future.
Flow computation when the capacities of the edges are time-varying was also studied in  \cite{hamacher2003earliest}.
%\cite{hamacher2003earliest,hoppe1995efficient,baumann2009earliest}

% transit times at edges are assumed to be constant 
%define a model, on which the edges expire at certain time. They take into consideration the maximum flow that transfers from the source to destination. Also, there is a \textit{delay} when the flow carries among the vertices because of the availability of the edges in network (the edges last a couple of days and the flow cannot carries directly from source node to destination node).

As opposed to all temporal flow computation problems studied in previous work
\cite{Skutella2008AnIT,DBLP:conf/ciac/AkridaCGKS17}
we do not consider networks where edges have capacities (variable or constant), but edges having sequences of instantaneous interactions with flow, which take place at specific timestamps.
Our objective is to compute the flow from a given source to a given sink vertex considering all interactions on the edges.
Still, as we show in Section \ref{sec:lp}, the maximum flow version of our problem is equivalent to the problem formulated and studied in \cite{DBLP:conf/ciac/AkridaCGKS17}, if we consider the interactions as ephemeral capacities of the edges.
Besides showing this equivalence, in Section \ref{sec:maxflowcomp}, we propose novel graph preprocessing and simplification techniques that greatly reduce the worst-case cost of maximum flow computation in practice. 

%in temporal interaction networks that we study in this paper and previous work on temporal flow computation
%\cite{Skutella2008AnIT,DBLP:conf/ciac/AkridaCGKS17}
%is that in our networks edges do not have capacities,
%but there are certain (instantaneous) interactions on 
%Although our flow computation definition shares similarities to the temporal flow of \cite{DBLP:conf/ciac/AkridaCGKS17} and in general to static and dynamic flow computation \cite{Skutella2008AnIT}, our problem definitions are very different.
%Our goal is not to measure the maximum flow that can be transferred from $s$ to $t$ given capacity constraints, transit rates and availability of edges, but to measure the {\em actual} total flow that is transferred from $s$ to $t$, given the factual transactions at the edges, the order of transactions and the assumption that nodes have infinite buffering capabilities. In addition, we ignore flow transfer quantities from a node $v$ that do not originally come from $s$, either because they were temporally before the incoming flows to $v$ from $s$ or because they exceed the total flow that entered $v$ from $s$. In addition, our algorithm is totally different (and asymptotically faster) than algorithms used to solve max flow problems.
%know the amount of flow that transfers from the first to the last node as well as the time of the interaction happens and (ii) in our model we consider the edges as consecutives.

Flow computation in temporal networks is also related to similar problems, but with a quite different formulation and goal. For instance, Kumar et al. \cite{DBLP:conf/edbt/0002C17} study the identification of interaction sequences between nodes forming paths in the network, which model potential pathways for information spread.
%Our work is similar with \cite{DBLP:conf/edbt/0002C17} but in our case (i) we take into account the flow on edges and (ii) we know already the amount of flow that transfers from source node to destination node.

%\fix{add more citations for flow in temporal networks}

\subsection{Network Patterns}
A number of works in the literature study the enumeration of graph patterns in static and temporal networks.
One of the earliest works is on motifs search \cite{Milo02networkmotifs:}.
Motifs are small patterns that repeat themselves in a network much more frequently than expected.
% . They defined motifs as \textit{``pattern of interconnections occurring in complex networks at numbers that are
%  significantly higher than those in randomized networks''}. Their aim was to identify them in large networks disregarding the time on edges (actually their definition for motifs refers only on static networks).
Paranjape et al. \cite{DBLP:conf/wsdm/ParanjapeBL17}
define motifs in temporal networks \cite{DBLP:journals/corr/abs-1108-1780}, by extending the definition of
\cite{Milo02networkmotifs:}, to consider the temporal information of the interactions between the graph vertices.
Specifically, such motifs are small connected graphs
whose edges are temporally ordered.
An instance of a motif is a sequence of interactions which have the structure of the motif and respect the time order of the motif's edges.
In addition, the time difference between the first and the last interaction in the instance should not exceed a maximum threshold.
The objective of motif search is to count the instances of one or more motifs in a large network.
%Specifically, an instance 
% . in temporal networks in the same way that patterns defined and studied extensively
%In, 
%A constraint  $\delta$ bounds the time difference between the edges  
%Also, in their definition, they introduce a constraint 
%$\delta$ to define the time-difference between the first and last edge ( the time-difference should be at most $\delta$). They propose an fast algorithm for counting all the motif instances in a graph.

Kosyfaki et al. \cite{KosyfakiMPT19}  defined and studied the enumeration of \textit{flow motifs} in interaction networks, considering both the time and the flow on the interactions.
% In their definition, the edges are labeled with a number $\ell(e)$  for identifying the total order of the edges. This is necessary in order to show the direction of the flow in the instance.
Such motifs come with two constraints: the maximum possible duration a motif instance and the minimum possible flow of the motif.
%An algorithm was proposed for enumerating flow motifs in large graphs.
Although we also study the enumeration of flow patterns in Section \ref{sec:algorithm}, (i) our flow computation model is very different compared to the one in \cite{KosyfakiMPT19}, as we consider maximum flow computation and also allow time-interleaving sequences of interactions, (ii) we study patterns that are not limited to simple paths, (iii) we propose precomputation approaches for pattern enumeration.
% Our definition has some similarities to flow motifs \cite{KosyfakiMPT19}, however in our case (i) our algorithm is linear as opposed to their algorithm which is exponential to the edges of the motif (ii) we have a totally different definition for flow computation.

Pattern matching and enumeration in general graphs and temporal networks is a well-studied problem \cite{gallagher2006matching,DBLP:conf/icde/SemertzidisP16,
  DBLP:journals/pvldb/SunWWSL12, DBLP:conf/icde/RanuS09,DBLP:journals/corr/RedmondC16,zou2009distance}.
Sun et al. study the problem of pattern matching in large networks \cite{DBLP:journals/pvldb/SunWWSL12}. They propose \textit{STwig}, an
% efficient algorithm for finding pattern instances. The
algorithm that combines graph exploration and joining intermediate results.
Moreover, their algorithm can adapted to work in parallel.
An older work \cite{DBLP:conf/icde/ChengYDYW08}
formulates and solves the pattern matching problem by joining the edges of the graph in a systematic way.
%is. The authors introduce an algorithm for finding subgraphs in large networks based on \textit{R-Joins} and \textit{R-semi Joins}.
The enumeration of cyclic patterns in a temporal network was recently studied in \cite{DBLP:journals/pvldb/KumarC18}.
Z\"{u}fle et al. \cite{DBLP:conf/edbt/ZufleREF18} study the temporal relations between entities in social networks. For this purpose, they  enumerate temporal patterns with the help of a data structure that indexes  small pattern instances.
Mining frequent patterns in static and temporal networks has also been studied in previous work \cite{DBLP:conf/sigmod/GurukarRR15,kuramochi2001frequent,kuramochi2005finding}.
%None of the works consider the concept of flow in patterns.

For graphs, where vertices (and/or edges) of the graph and the patterns are labeled, pattern matching is relatively easy, as vertex labels (or small subgraphs) can be indexed and search/join algorithms can be used to accelerate search. The flow pattern search problem that we study in Section \ref{sec:algorithm} is more challenging, because there are no constraints as to which vertices of the graph can match the vertices of a pattern. In addition, there is no previous work on enumerating pattern instances {\em and} their maximum flows, i.e., the problem that we study in Section \ref{sec:algorithm}.
%In this paper, we focus on the problem of computing pattern instances in interaction networks and their maximum flows and show how it can be solved efficiently by precomputing instances of small patterns and their flows.
%also study the for each pattern instance, 

\section{Definitions}\label{sec:def}
%Our goal is the enumeration of flow pattern instances in a large interaction network. 
%Before we model the problem of flow computation,
In
this section, we define basic concepts and summarize the most frequently
used notation. We begin by formally defining an interaction network.

\begin{definition}[Interaction Network]\label{def:network}
  An interaction \linebreak network is a directed graph $G(V,E)$. For each edge
  $e=(v,u)$ of the network, there is a sequence
  $e_S=\{(t_1,q_1), (t_2,q_2),\dots\}$ of interactions from
  node $v$ to node $u$. Each interaction $(t_i,q_i)$ has a quantity
  $q_i$, which is transferred from $v$ to $u$ at timestamp $t_i$.
  %For any two timestamps $t_i,t_j$, $t_i<t_j$ iff $i<j$. 
\end{definition}

Figure \ref{fig:example}(a) shows an example of a
network, where edges are annotated with the corresponding
sequences of interactions.

In practice, an data analyst would be interested in measuring the total
flow from a specific {\em source} vertex of the network to
a specific {\em sink} vertex of the network. The source and the sink
might coincide.
In addition, the analyst might only want to
include certain vertices and edges in the subgraph
for which the flow is to be measured.

One way to select interesting subgraphs on which the flow should be measured
is by specifying a {\em network pattern}, i.e., the structure that the
interesting subgraphs should conform to, and identify their
{\em instances} in the interaction network.
We now provide definitions for a network pattern and its instances.
%and
%then provide some intuition about these definitions.

%Note that
%the quantities 
%in the transactions are measured in $\$$.

%Considering
%time interval $I=[2,5]$, the $I$-restricted edge $(u_1,u_2)^I$ has
%$(u_1,u_2)^I_T=\{(2,\$5),(4,\$3)\}$.

\begin{figure}[tbh]
  \centering
  \small
  \begin{tabular}{@{}c c c@{}}
  \includegraphics[width=0.34\columnwidth]{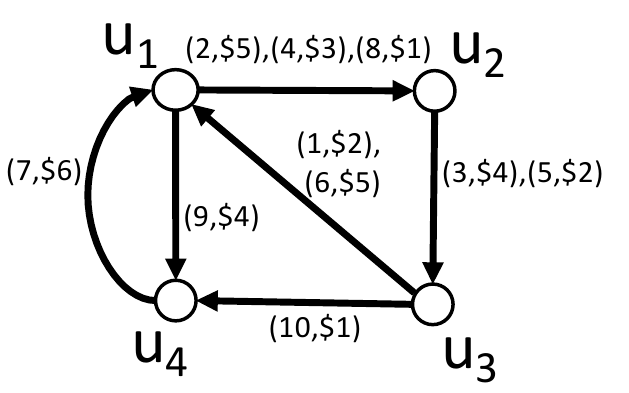}&
  \includegraphics[width=0.06\columnwidth]{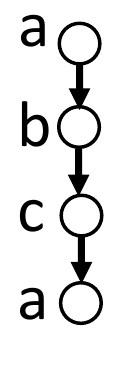}&
  \includegraphics[width=0.27\columnwidth]{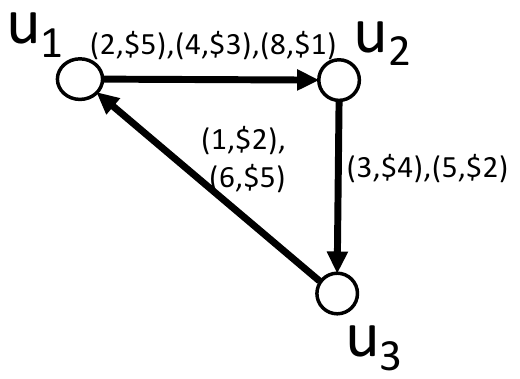}\\
    (a) interaction network&
         (b) pattern&
              (c) instance\\
  \end{tabular}
  \caption{Network, pattern, and instance}
  \label{fig:example}
\end{figure}

\begin{definition}[Network Pattern]\label{def:pattern}
  A network pattern \linebreak $G_P(V_P,E_P)$ is a directed acyclic graph, where
  each vertex $v\in V_P$ has a label $\ell(v)$.
\end{definition}

\begin{definition}[Instance]\label{def:match}
  An instance of pattern $G_P$ in graph $G$ is a subgraph $G_M(V_M,E_M)$
  of $G_T$, such that
  \begin{itemize}
  \item   there is a surjection $\mu: V_P\to V_M$
  from the vertex set $V_P$ of the pattern $G_P$ to the vertex set
  $V_M$ of $G_M$;
  \item for two vertices
    $v,u$ of $G_P$,
    % map to the same vertex of $M_T$, i.e.,
    $\mu(v)=\mu(u)$ iff
  $\ell(v)=\ell(u)$;
  \item  $(v,u)\in E_P$ iff $(\mu(v),\mu(u))\in E_M$.
  \end{itemize}
\end{definition}

%\begin{figure}
%\centering
%\includegraphics[width=1.15\columnwidth]{patterns.pdf}
%\caption{Examples of patterns.}
% \label{fig:patterns_ex} \vspace{-.5em}
%\end{figure}

We assume that the network $G$, in which we search for pattern instances 
is not labeled. The labels on the vertices of a pattern are only used to indicate
that pattern vertices having the same labels
should be mapped to the same graph vertex in a pattern instance.
Continuing the example of Figure \ref{fig:example}, consider the
network pattern $G_P$ shown in Figure \ref{fig:example}(b) which
includes 4 nodes connected in a chain. Since the first and
the last vertex of $G_P$ have the same label, this pattern corresponds to a
cyclic transaction (i.e., $a$ transfers some quantity to $b$, then $b$
transfers to $c$, then $c$ transfers to $a$).  Figure
\ref{fig:example}(c) shows an instance of this pattern, which
is a subgraph of the interaction network that satisfies Definition
\ref{def:match} (i.e., $a$ is mapped to $u_1$, $b$ to $u_2$, and $c$ to
$u_3$).

%we present our methodology, we provide a detailed definition of flow
%patterns and their (structural) instances in a graph. In addition, we
%define the problem of computing the flow in a given pattern instance.
In the next section, we define and study the
problem of computing the total quantity that flows throughout
an interaction network or a subgraph thereof,
from a given source to a given sink vertex, before studying the
problem of enumerating network patterns and their flows in Section \ref{sec:algorithm}.
Table \ref{table:notations} summarizes the notation used frequently in
the paper.
%\todo{+++}

%\makebox[1.3 \textwidth][c]{ 
\begin{table}[ht]
\caption{Table of notations}\label{table:notations}
\centering
\scriptsize
\begin{tabular}{@{}|@{~}c@{~}|@{~}c@{~}|@{}}
\hline
Notations &Description\\ % 
\hline  
  $G(V,E)$ & input graph \\
%$e.src$ & source vertex of edge $e$\\
%$e.dest$ & destination vertex of edge $e$\\
$(t_i,q_i)$ & an interaction with quantity $q_i$ at time $t_i$\\ 
$src_i$ ($dest_i$) & source (destination) vertex of interaction $(t_i,q_i)$\\ 
$e_S=\{(t_1,q_1), (t_2,q_2),\dots\}$ & 	 sequence of interactions on edge $e$\\
%$e^I$, $e^I_T$ & time-restricted edge and transactions sequence\\
  $B_v$ & total quantity buffered at node $v$\\
  $B^{t_i}_v$ & buffer at node $v$ by time $t_i$\\
  $G_P(V_P,E_P)$ & network pattern \\
 $G_M(V_M,E_M)$ & instance a network pattern\\
  $\mu(v)$ &  vertex of $V_M$ mapped to vertex $v\in V_P$\\
  
%$G^I_M(V^I_M,E^I_M)$  & time-restricted structural match \\

\hline
\end{tabular}
\end{table} 
\section{Flow Computation}\label{sec:flowcomp}
%Before we define the problem of flow pattern enumeration in
%an interaction network,
Here, we focus on the problem of flow computation
from a given source to a given sink vertex,
in an interaction (sub-)network.
%which forms a {\em directed
%acyclic graph} (DAG).
%\todo{justify why we focus on DAGs}
%, where vertices correspond to entities and edges
%correspond to sequences of interactions between the entities.
Figure \ref{fig:flowcompex} illustrates an example of such a network,
consisting of four vertices and five edges. Each edge has a sequence
$\{(t_i,q_i)\}$ of interactions; in this example, each sequence has
only one interaction. 
%\fix{add figure and describe it}.

\begin{figure}[tbh]
  \centering
  \small
  \includegraphics[width=0.50\columnwidth]{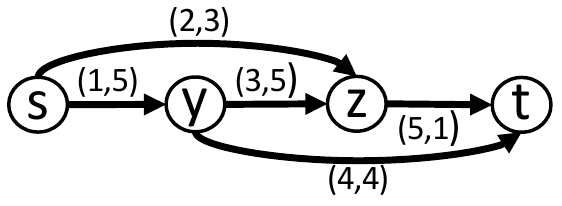}
  \caption{Example of a graph}
  \label{fig:flowcompex}
\end{figure}

We consider two definitions of the flow computation problem and how
they relate to each other.
First, in {\em greedy flow computation}, we take interactions in
order of time and assume that every
interaction $(t_i,q_i)$ greedily transfers the maximum possible quantity to the
target vertex,
given the quantity accumulated until time $t_i$ to the source vertex of the
interaction.
For example, in the graph of Figure \ref{fig:flowcompex}, interaction
$(t_i,q_i)=(3,5)$ on edge $(y,z)$ transfers a quantity of $q_i=5$ from $y$ to $z$,
because $y$ has received $5$ due to interaction $(t_j,q_j)=(1,5)$ on edge
$(s,y)$ which happened before ($t_j<t_i$).
%\fix{add example}.
On the other hand, in {\em maximum flow computation}, we
consider the case where a vertex may {\em reserve} some quantity
for future interactions, if that could maximize the maximum overall flow that
can be transferred throughout the DAG.
For example, in Figure \ref{fig:flowcompex}, interaction
$(t_i,q_i)=(3,5)$ on edge $(y,z)$ may transfer any quantity in $[0,5]$, since
vertex $y$ has accumulated $5$ units, by time $3$.
%, due to interaction $(t_j,q_j)=(1,5)$ on edge
%$(s,y)$ which happened before ($t_j<t_i$).
%\fix{add example}.
Both definitions comply to the principle that
a interaction $(t_i,q_i)$ on an edge $(v,u)$
cannot transfer a larger quantity than what
the source vertex $v$ has received from its incoming
interactions before time $t_i$
and {\em was not yet transferred} via its outgoing
interactions before time $t_i$.

In both definitions of flow computation,
we consider {\em connected} graphs which have just one {\em source} node
(with no incoming edges) and just one {\em sink} node with no outgoing
edges (like the graph of Figure \ref{fig:flowcompex}).
Our methods and algorithms can easily be extended for graphs with
multiple
sources. In this case, we can add a synthetic
source vertex $s$ and an edge from $s$ to each original source.
Similarly, if there are multiple sinks, we can add a synthetic
sink vertex $t$ and an edge from each original sink to the synthetic one.
Each of the outgoing edges from the synthetic source are given a single
interaction with the smallest possible timestamp and an infinite
quantity (in order for the original sources to be able to
transfer any quantity via their outgoing edges).
Each of the incoming edges of the synthetic sink are given a single
interaction with the largest possible timestamp and an infinite
quantity (in order for the original sinks to be able to
absorb any quantity via their incoming edges; these are eventually
accumulated at the synthetic sink).
Figure \ref{fig:ex_ss} shows an example of a graph before and after
adding the source and the sink.
In the rest of the paper, we assume that each input graph to our
flow computation problem
is connected and has a
single source vertex $s$ with no incoming edges
and a single sink vertex $t$ with no outgoing edges.
The objective is to compute the flow from $s$ to $t$.
%\fix{Fix figure \ref{fig:ex_ss}(b), change $a$ to $s$ and $f$ to $t$}

\begin{figure}[tbh]
  \centering
  \small
  \begin{tabular}{@{}c c@{}}
  \includegraphics[width=0.2\columnwidth]{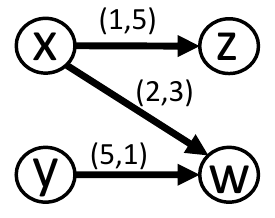}&
  \includegraphics[width=0.45\columnwidth]{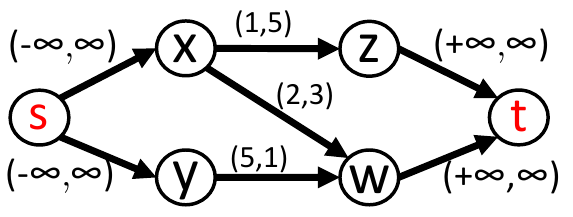}\\
    (a) Example of a DAG&
         (b) Additional source and sink\\
  \end{tabular}
  \caption{Example of synthetic source and sink addition}
  \label{fig:ex_ss}
\end{figure}

\subsection{Greedy flow computation}\label{sec:flow:greedy}
In this section, we define and solve the greedy flow computation
problem in a graph $G$.
To compute the flow $f(G)$ throughout $G$, 
we assume that each vertex $v\in G$
keeps, in a buffer $B_v$, the total quantity received from its incoming
interactions.
We denote by $B^{t_i}_v$ the value of the buffer by time $t_i$.
The source vertex $s$ of $G$ is assumed to constantly have
{\em infinite quantity} in its buffer, i.e.,
$B^{t_i}_s=\infty, \forall t$.
This means that for each interaction $(t_i,q_i)$ that comes out of
$s$, the entire quantity $q_i$ is transferred to the buffer of the destination vertex.
Before the temporally first interaction in
$G$, the buffers of all nodes (except for the source) are $0$.
The flow computation process considers all interactions at the edges
of $E$ {\em in order of time}.
Each interaction, say from vertex $v$ to vertex $u$
greedily transfers
the maximum possible quantity from $B_v$ to $B_u$.
We do not set a bound on how much a node can buffer and buffered
quantities do not expire.
Formally, flow computation is based on the following
definition of greedy flow transfer.

\begin{definition}[Greedy Flow Transfer]\label{def:greedyflowtransfer}
  As an effect of interaction $(t_i,q_i)$ on edge $(v,u)\in G$, $v$
  transfers to $u$ at time $t_i$, a quantity $q=\min\{q_i,B^{t_i}_v\}$,
  where $B^{t_i}_v$ is the total quantity buffered in $v$ by time
  $t_i$.
  As a result of the flow transfer, $B^{t_i}_v$ is
  reduced by $q$ and $B^{t_i}_u$ is increased by $q$.
\end{definition}

In simple words, as a result of a interaction, a node transfers as
much as possible from its buffered quantity via the interaction. 
After
the last interaction in $G$, the total quantity buffered at the 
sink vertex of $G$ is the flow $f(G)$. Formally,

\begin{definition}[Greedy flow of graph $G$]\label{def:flow}
The flow $f(G)$ of a graph $G$ is the total quantity
buffered at the sink of $G$ after processing all interactions in
$G$ in temporal order and applying
for each interaction Definition \ref{def:greedyflowtransfer}
to update the buffered quantities of nodes.
\end{definition}

%\fix{fix this}
Table \ref{tab:ex_flow} shows the steps of computing $f(G)$ of the
graph shown in \ref{fig:flowcompex}. Since $s$ is the source node of
the graph, $B_s$ is always $\infty$.
In addition, $t$ is the sink, hence, $f(G)$ will be equal to $B_t$,
after processing all interactions.
The first column shows the currently examined
interaction, the second column the edge where it belongs and the last
four columns the changes in the buffers of the vertices after the
interaction is processed.
In the beginning, $B_s=\infty$ and the buffers of
all other vertices are $0$.
The temporally first interaction
$(1,5)$ on edge $(s,y)$ transfers $5$ units from $B_s$ to $B_y$.
Then, $(2,3)$
on edge $(s,z)$ transfers $3$ units from $B_s$ to $B_z$.
Then, $(3,5)$
on edge $(y,z)$ transfers $\min\{B_y,5\}=5$ units from $B_y$ to $B_z$,
which results in $B_y=0$ and $B_z=8$.
Interaction $(4,4)$ on edge $(y,t)$ transfers no units, as
$\min\{B_y,4\}=0$.
Finally, interaction $(5,1)$ on edge $(z,t)$ transfers
$\min\{B_z,1\}=1$ units from $B_z$ to $B_t$ and the total flow of the
DAG is considered to be $B_t=1$. 

%As discussed, the temporally first
%interaction $(1,\$2)$ on $(u_3,u_1)$ has no effect, as $B_{u_3}=0$ at
%time $1$. After processing the last interaction, the buffer
%$B_{u_1(t)}$ of the sink node contains the flow $f(G_M)$. 

\begin{table}[tbh]
  \caption{Example of  greedy flow computation}
  \label{tab:ex_flow}
  \centering
  \small
  \begin{tabular}{@{}|c |c |c |c |c |c|@{}}
    \hline
    $(t_i,q_i)$ & $(v,u)$& $B_{s}$& $B_{y}$& $B_{z}$& $B_{t}$\\\hline
    $(1,5)$ & $(s,y)$& $\infty$& $5$& $0$& $0$\\
    $(2,3)$ & $(s,z)$& $\infty$& $5$& $3$& $0$\\
    $(3,5)$ & $(y,z)$& $\infty$& $0$& $8$& $0$\\
    $(4,4)$ & $(y,t)$& $\infty$& $0$& $8$& $0$\\
    $(5,1)$ & $(z,t)$& $\infty$& $0$& $7$& $1$\\
   
    \hline
  \end{tabular}
\end{table}

\noindent
{\bf Complexity analysis.}
It is easy to show that the flow $f(G)$ of a graph $G$ can be
computed in time linear to the number of interactions on the edges of
$G$, assuming that these can be accessed in order of time.
This is due to the fact that
each interaction causes the update of at most two vertex buffers, hence,
processing an interaction takes constant time.
%If the interactions are not ordered by time, then we have to bear the
%cost of sorting them

\subsection{Maximum Flow computation}\label{sec:maxflowcomp}

%\begin{figure}[tbh]
%  \centering
%  \small
%  \begin{tabular}{@{}c c c@{}}
%  \includegraphics[width=0.34\columnwidth]{graph_m.png}&
%  \includegraphics[width=0.30\columnwidth]{pattern_m.png}&
%  \includegraphics[width=0.24\columnwidth]{structural_m.png}\\
%    (a) interaction network&
%         (b) pattern&
%              (c) structural match\\
%  \end{tabular}
%  \caption{Network, pattern, and structural match}
%  \label{fig:example2}
%\end{figure}

The flow transfer definition of the previous section (Definition
\ref{def:greedyflowtransfer}) does not consider the case where,
as a result of an interaction on edge $(v,u)$,
$v$  does not transfer the maximum possible quantity
to $u$, but {\em reserves} quantity for future interactions.
As a result, the quantity $f(G)$ computed by Definition 
\ref{def:flow}
might not be the maximum possible.
%(hence, the name ``greedy flow transfer").
To illustrate this, consider again the graph of 
Figure \ref{fig:flowcompex}.
As shown in Table \ref{tab:ex_flow}, due to interaction $(3,5)$,
vertex $y$ transfers all its buffered quantity (i.e., $5$) to $z$,
hence $B_y$ becomes 0 and $B_z$ becomes 8.
As a result, at the temporally
next interaction $(4,4)$, which is on edge $(y,t)$, $y$ cannot transfer any
quantity to $t$ because $B_y=0$ at that time. Had
$y$ transferred to $z$ a quantity of $1$ unit (instead of $5$)
at the interaction  $(3,5)$, $y$
would have saved $4$ units in its buffer to use at interaction
$(4,4)$.
This change maximizes the total flow that is transferred throughout
the graph, as shown in Table \ref{tab:ex_flow2}. In general, the flow
computed by the greedy algorithm can be arbitrarily smaller than the
maximum possible flow.
 
  \begin{table}[tbh]
  \caption{Example of maximum flow computation}
  \label{tab:ex_flow2}
  \centering
  \small
  \begin{tabular}{@{}|c |c |c |c |c |c|@{}}
    \hline
    $(t_i,q_i)$ & $(v,u)$& $B_{s}$& $B_{y}$& $B_{z}$& $B_{t}$\\\hline
    $(1,5)$ & $(s,y)$& $\infty$& $5$& $0$& $0$\\
    $(2,3)$ & $(s,z)$& $\infty$& $5$& $3$& $0$\\
    $(3,5)$ & $(y,z)$& $\infty$& $4$& $4$& $0$\\
    $(4,4)$ & $(y,t)$& $\infty$& $0$& $4$& $4$\\
    $(5,1)$ & $(z,t)$& $\infty$& $0$& $3$& $5$\\
   
    \hline
  \end{tabular}
\end{table}

Hence, assuming that vertices can transfer any portion of their
reserved quantity at an interaction, an interesting problem is finding
the maximum flow that can be transferred from the source to the sink
of the graph $G$. In this section, we analyze this problem and show that
it is equivalent to a maximum flow computation problem in temporal
graphs \cite{DBLP:conf/ciac/AkridaCGKS17},
which can be solved by linear programming (LP).
We show that for specific classes of graphs $G$ (e.g., chains), greedy flow
computation gives us the solution to the maximum flow problem.
In addition, to reduce the cost of 
maximum flow computation,
we propose a preprocessing approach, which eliminates interactions
(and possibly edges and vertices of the graph) which are guaranteed
do not affect the solution.
Finally, we present a graph simplification approach, which computes
part of the solution using the
greedy algorithm and, consequently,
reduce the overall cost  maximum flow computation.
%technique that can
%hierarchically divide
%the DAG into components, such that some components can be solved using
%the greedy algorithm and some using linear programming, reducing the
%overall cost of computing the maximum flow.
%using LP only. 

%\nikos{Explain max flow problem and show what is the max flow in
%  the example. Explain that for chain patterns max-flow is the same as
%  greedy flow. Explain how max flow can be computed using LP. Explain
%  how to decompose a pattern, in order to compute max flow, where for
%  some segments you run LP and for some segments you run Greedy.} 

%\subsubsection{Equivalence to maximum flow problem on temporal
%graphs}
\subsubsection{Formulation as an LP problem and equivalence to a known problem}\label{sec:lp}
We first formulate the maximum flow computation problem as a linear
programming (LP) problem.
The problem includes one variable $x_i$ for each interaction
$(t_i,q_i)$ at any edge.
Variable $x_i$ corresponds to the quantity that will be transferred
as a result of the interaction.
Since the transferred quantity cannot be negative and cannot exceed
$q_i$, we have:
\begin{equation}\label{eq:con1}
0\le x_i\le q_i
\end{equation}
For the special case, where the interaction originates from the source
vertex, we have $x_i=q_i$, since we assume that the source
has infinite buffer (i.e., reducing the units transferred from the
source vertex cannot increase the total quantity that reaches the
sink). Hence, the number of variables can be reduced to the number of
interactions that do not originate from the source.
In addition, we have the constraint that an interaction $(t_i,q_i)$ on
edge $(src_i,dest_i)$
cannot transfer more than the total incoming units to $src_i$ 
minus the total outgoing units from $src_i$, up to timestamp $t_i$:
\begin{equation}\label{eq:con2}
x_i\le \sum_{dest_j=src_i \land  t_j<t_i}x_j-\sum_{src_j=src_i \land  t_j<t_i}x_j
\end{equation}

The objective of the LP problem is to
find the values of all variables $x_i$, which will maximize the quantity
that will arrive at the sink vertex. Hence, the objective is:
\begin{equation}\label{eq:objective}
  \textrm{maximize~}
  \sum_{dest_i=sink(G)}{x_i}
\end{equation}

We will now show that our problem is
equivalent to the 
maximum flow computation problem in temporal
graphs, studied in \cite{DBLP:conf/ciac/AkridaCGKS17}.
Specifically, a temporal flow network, as defined in
\cite{DBLP:conf/ciac/AkridaCGKS17}, each edge has a capacity
and the edge contains a set of time moments during which
the edge can transfer flow up to its capacity (until the next time moment).
Assuming that infinite quantity is available at the source
vertex at time zero, the objective is to find the maximum total
flow
that can reach the sink vertex of the network
after all time moments of edge availabilities have passed.
%at the end of the
%time history that includes all the moments of edges' availabilities.
It is not hard to see that this problem is equivalent to our problem
if we set as time moments the times of the interactions and as
capacities the corresponding quantities $q_i$. As shown in
\cite{DBLP:conf/ciac/AkridaCGKS17}, the problem can be solved
in PTIME and can be converted to a classic max-flow computation
problem in static networks. In the equivalent static network,
for each time moment
of edge activity one edge is added linking versions of the
corresponding vertices. Hence, the complexity of the problem  is
quadratic to the total number of activity time moments on the edges.
Equivalently, computing the maximum flow on a temporal interaction
network (i.e., our problem) has quadratic cost to the number of
interactions on the edges. 

%\nikos{Explain equivalence here}

\subsubsection{Graphs for which the greedy algorithm solves the maximum flow
problem}\label{sec:greedyreduction}

Solving our problem directly using LP (or any other max-flow
algorithm) is not as efficient as applying the greedy algorithm, which
computes the flow in time linear to the number of
interactions. However, the greedy algorithm does not always compute
the maximum flow, as we have shown already.
%\nikos{Must theoretically prove that greedy solves the max flow
%  problem on chains. Must show this also for parallel edges.}
We will now show that for special cases of graphs, the greedy flow
computation algorithm indeed computes the maximum flow. This means that for
such networks, maximum flow computation can be done in time linear to
the number of interactions (assuming that these are sorted by time).

Chains are the first class of graphs where this applies.
A chain is a connected directed acyclic graph (DAG)
for which (i) the source node $s$ has just one
outgoing edge, (ii) the sink node $t$ has just one incoming edge, (iii)
every other node has only one incoming and only one outgoing edge.
In simple words, a chain is a DAG in which all edges form a single
path that connects all nodes.
For example, Figure \ref{fig:greedyopt}(a) shows a chain DAG
consisting of four nodes ($s,y,z,t$) and three edges having
in total 7 interactions on them.

\begin{figure}[tbh]
\centering
\small
\begin{tabular}{@{}c c @{}}
\includegraphics[width=0.40\columnwidth]{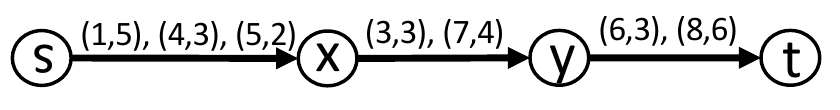}&
 \includegraphics[width=0.46\columnwidth]{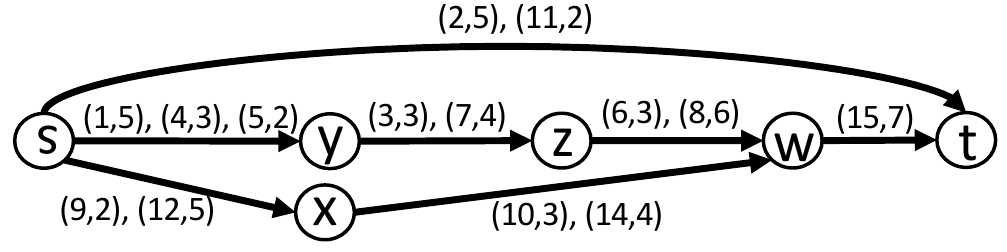}\\
(a) chain DAG&
(b) non-chain DAG\\
\end{tabular}
\caption{Maximum flow computation by greedy}
\label{fig:greedyopt}
\end{figure}

\begin{lemma}\label{lem:chains}
  If $G$ is a chain, the greedy algorithm computes the maximum flow in $G$. 
\end{lemma}
\begin{proof} [Sketch]
We will prove the lemma by induction. The lemma trivially holds for
the base case, when
$G$ is a simple edge $(s,t)$.
In this case, the greedy algorithm
%Since $B_v=\infty$, 
sends to buffer $B_t$  the total quantity from all
interactions on the edge $(s,t)$
(since $B_s=\infty$). Hence, $B_t$ has received the
maximum possible flow at
every time moment.
For the inductive step,
consider a chain $G$ with the last edge of the chain
being $(v,t)$.
We will assume that $B_v$ has received the maximum possible
flow (from the previous vertices of the chain)
at any timestamp and prove that
$B_t$ will receive
the maximum possible flow from its incoming edges (i.e., from $v$)
at any timestamp.
Assume that due to an interaction $(t_i,q_i)$ on edge $(v,t)$, instead of
applying the greedy algorithm to transfer the maximum possible flow
$q=\min\{q_i,B^{t_i}_v\}$ from $B^{t_i}_v$ to $B^{t_i}_t$, we transfer
a smaller quantity $q'<q$. We can easily prove that, as a result of
this change, the accumulated flow $B_t$ at $t$, after processing all
interactions cannot increase. The reason is that $t$ receives flow
only from $v$, hence, flow reservation by $v$ cannot increase the
total
flow which will be sent from $v$ to $t$. 
\end{proof}

We can generalize Lemma \ref{lem:chains} and show that the greedy
algorithm computes the maximum flow for DAGs where
the source is the only vertex
which may have more than one outgoing edges.

\begin{lemma}\label{lem:greedyoptimal}
Let $G(V,E)$ be a DAG where the source vertex is $s$ and the sink vertex is
$t$.
The greedy algorithm computes the maximum flow throughout $G$
if for every vertex $v \in V\backslash \{s,t\}$, $v$ has exactly one
outgoing edge.
\end{lemma}

\begin{proof} [Sketch]
Similarly to the proof of Lemma \ref{lem:chains}, assume that
a vertex
$v \in V\backslash \{s,t\}$ having outgoing edge $(v,u)$
does not transfer the maximum possible flow as a result of an
interaction $(t_i,q_i)$ on $(v,u)$, but retains some quantity.
This cannot increase the total quantity that reaches $u$ (and eventually $t$)
via $v$ in future interactions  $(t_j,q_j), t_j>t_i$ stemming from $v$
because $(v,u)$ is the only outgoing edge from $v$ and $t$ can be
reached from $v$ only via $u$.
In addition, there is no benefit in retaining quantities at the source
vertex $s$.
Hence, greedily transferring the maximum possible quantity
at every interaction, results in accumulating the maximum flow at the
sink $t$.
\end{proof}

Figure \ref{fig:greedyopt}(b) shows an example, where the Greedy
algorithm computes the maximum flow ($=14$). Note that all vertices except for
the source $s$ and the sink $t$ have just one outgoing edge.
Checking whether the graph satisfies this condition costs just $O(V)$
time, i.e., examining the out-degree of each vertex.
On the other hand, we can easily construct examples of graphs that
do not satisfy this condition and for which Greedy does not compute the
maximum flow (like the graph in Figure \ref{fig:flowcompex}).

\subsubsection{Graph preprocessing}\label{sec:flowcomp:prepdag}
Before applying LP to compute the maximum flow on a DAG for which the
Greedy algorithm is not guaranteed to find the maximum flow,
i.e., a DAG that does not satisfy the condition of
Lemma \ref{lem:greedyoptimal}, we can reduce the
complexity of the problem by removing interactions that do not affect
the solution. 
For example, consider the pattern instance
of Figure \ref{fig:example}(c) and the last edge $(u_3,u_1)$ of this
DAG. On this edge, there is an interaction $(1,\$2)$ which obviously
does not account in the flow computation and can be ignored. The
reason is that the timestamp of this interaction is smaller than all
the timestamps of all interactions that enter $u_3$, i.e., the source
node of interaction $(1,\$2)$. In simple words, it is impossible for
$u_3$ to transfer $\$2$ at timestamp 1 to $u_1$, because by that time
it is not possible to have received any money from its incoming
interactions. Removing interactions can be crucial to the performance
of LP because there are as many variables
as the number of interactions in the DAG (except those originating
from the source vertex). 

Hence, based on the observation above, before applying LP, we
perform a {\em preprocessing step} on the DAG $G$, where we eliminate interactions
that cannot contribute to the maximum flow. Specifically, we consider
all vertices of $G$ in a {\em topological order} and for each vertex,
which is not the source or the sink of the DAG, we
examine its outgoing edges and remove from them all interactions with
a smaller  timestamp than the smallest incoming timestamp to the
vertex. The reason of examining the vertices in a topological order is
that the deletion of an interaction may trigger the deletion of
interactions in edges that follow. Examining the vertices in this
order guarantees updating the graph by a single pass over its
vertices.

For example, consider the DAG $G_1$ shown in Figure
\ref{fig:dagpre}(a).
%A topological order of the graph's vertices is \linebreak $\{a,b,c,d,e\}$.
To preprocess $G_1$, we consider its vertices
a topological order, i.e., $\{s,x,y,z,t\}$.
For each vertex having both incoming and outgoing edges
we attempt to
delete transactions from its outgoing edges.%
\footnote{We cannot eliminate any interactions from the source vertex
  of the DAG.}
The first such vertex is
$x$. First, we find the minumum timestamp of any incoming interaction
to $x$, which is $5$. Then, we examine the interactions on outgoing edge
$(x,y)$. From them, interaction $(2,7)$ is deleted because $2<5$.
Then, we examine the interactions on outgoing edge
$(x,z)$. From them, interaction $(1,2)$ is deleted because $1<5$.
We move on to vertex $y$. The minimum timestamp of any interaction
entering $y$ is now $9$ (recall that interaction $(2,7)$ on edge
$(x,y)$ has been deleted). This causes interaction $(3,3)$ on outgoing
edge $(y,t)$ from $y$ to be deleted.
Finally, we move on to vertex $z$; the incoming interaction
to $z$ with the minimum timestamp is $(10,5)$.
This causes interaction  $(4,2)$ on outgoing
edge $(z,t)$ from $z$ to be deleted.
The DAG after preprocessing is shown in Figure \ref{fig:dagpre}(b).

\begin{figure}[tbh]
\centering
\small
\begin{tabular}{@{}c c @{}}
\includegraphics[width=0.48\columnwidth]{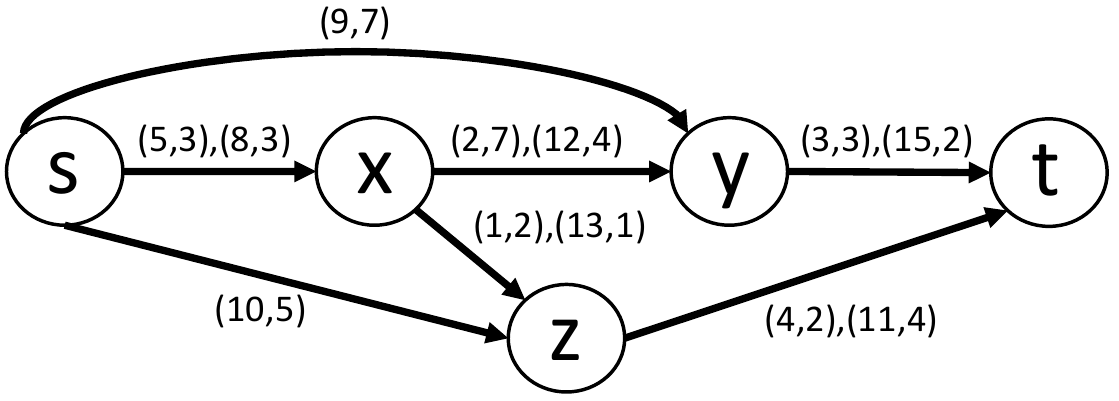}&
\includegraphics[width=0.48\columnwidth]{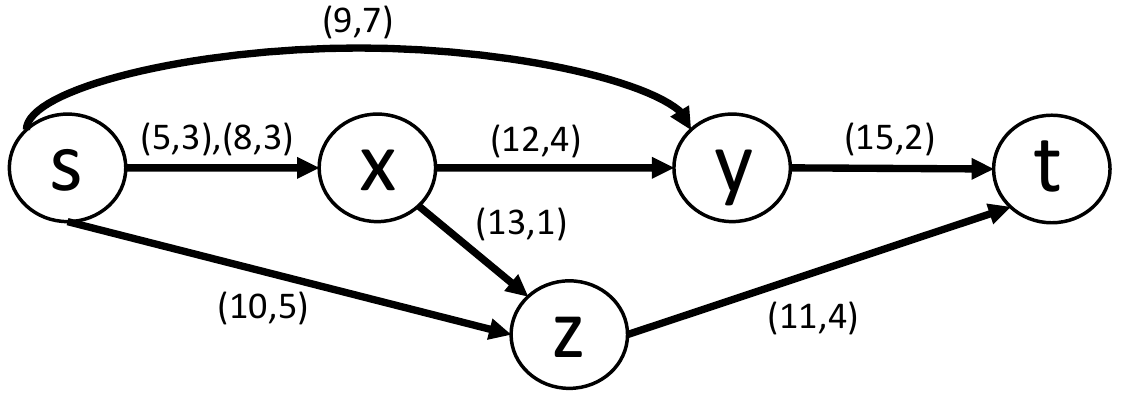}\\
(a) DAG $G_1$ before&
(b) DAG $G_1$ after\\
\includegraphics[width=0.48\columnwidth]{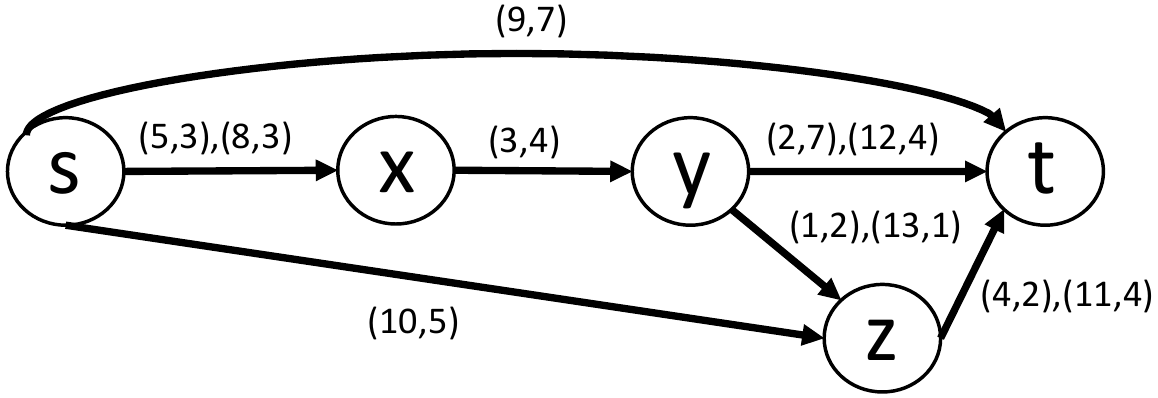}&
\includegraphics[width=0.48\columnwidth]{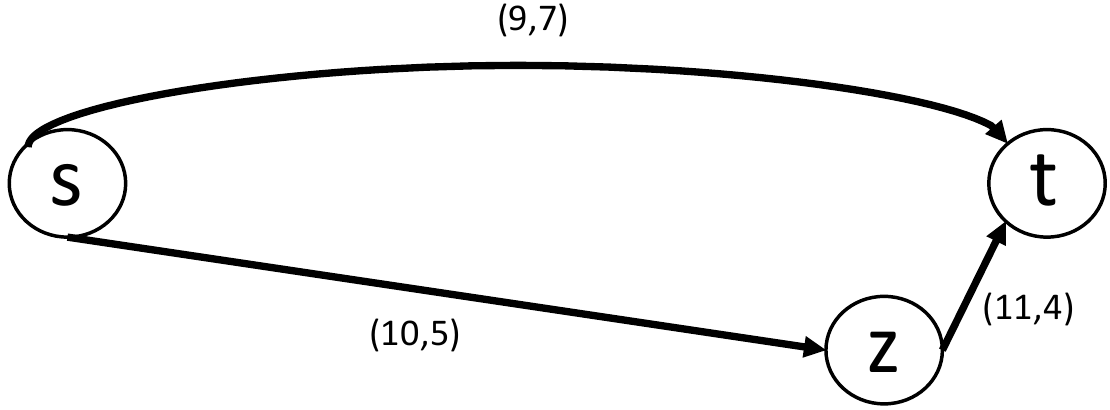}\\
(a) DAG $G_2$ before&
(b) DAG $G_2$ after\\
\end{tabular}
\caption{DAG preprocessing examples}
\label{fig:dagpre}
\end{figure}

The DAG preprocessing procedure described above may cause all interactions on an
edge to be deleted. In this case, the edge cannot transfer any
flow and hence should be deleted.
The deletions of edges may result in deletion of nodes and, in turn,
trigger the deletion of other edges and nodes.
We will now see how the deletion of edges can trigger the
simplification of the graph, which can greatly reduce the cost of
maximum flow computation.

The deletion of an edge $(v,u)$ may
have two effects: (i) the number of incoming edges to vertex $u$ becomes 0,
(ii) the number of outgoing edges from vertex $v$ becomes 0.
As described above,
the deletion of edge $(v,u)$ may happen when we examine vertex
$u$. Hence, case (i) can be handled when we examine vertex $u$, which
follows $v$ in the topological order. Specifically, if the currently
examined vertex $u$ has no incoming edges (the DAG's source vertex is not
examined, hence it is an exception here), then this means that no
quantity from the source of the DAG can flow through $u$
to the sink vertex of the DAG. Hence, {\em $u$ and all its outgoing edges
should be removed from the DAG}. If the removal of an edge $(u,u')$ makes its
destination vertex $u'$ to have no incoming edges, then this outcome
will be handled when $u'$ will be examined ($u'$ must follow $u$ in
the topological order).

If, after the deletion of an edge $(v,u)$, case (ii) applies, i.e.,
$v$ has no more outgoing edges, then this means that no flow can reach
the sink of the DAG via $v$.
Hence, $v$ and all its incoming edges should be deleted.
The deletion of an incoming edge $(w,v)$
may cause vertex $w$
to have no outgoing edges, in which case $w$ should also be deleted.
The deletion of $w$ should be done immediately, because $w$
precedes $v$ in the topological order and will not be examined later.
The deletion may trigger the deletion of other nodes and edges recursively.

Figures \ref{fig:dagpre}(c) and \ref{fig:dagpre}(d) show an example of
a DAG $G_2$ before and after preprocessing.
The vertices are examined in topological order $\{s,x,y,z,t\}$. Since
$s$ is the source and $t$ is the sink, only vertices $\{x,y,z\}$ are
examined in this order. We first examine $x$ and remove the single interaction
$(3,4)$ from its outgoing edge $(x,y)$, since $3<\min\{5,8\}$. This causes
edge $(x,y)$ to be deleted, which makes $x$ having no outgoing
edges. Hence, $x$ and all its incoming edges should be deleted as
well.
The next vertex to be examined is $y$, which has no incoming edges
(since edge $(x,y)$ has been deleted). Hence, $y$ and its outgoing edges
are deleted. The next vertex to be examined is $z$ and interaction
$(4,2)$ is removed from edge $(z,t)$. The final graph is shown in
\ref{fig:dagpre}(d). Note that this graph is soluble by Greedy; hence,
if DAG preprocessing removes edges from the graph, we apply again the
condition of Lemma \ref{lem:greedyoptimal} to check if the resulting
DAG is soluble by Greedy.

A pseudocode for this DAG preprocessing procedure is Algorithm
\ref{algo:dagpre}.
The algorithm can significantly reduce the size of the problem, by
removing interactions, edges, and nodes. 
In the case where the source or the sink of the DAG is removed,
the DAG has 0 flow,
which means that we can avoid running the flow computation
algorithm.
The source vertex can be deleted in the case where the deletion of a
vertex propagates upwards until the source.
The sink node can be deleted if all its incoming edges are deleted.
In any case, after preprocessing, the resulting DAG should be
connected and all vertices which are not the source and the sink
should have at least one incoming and at least one outgoing edge.
The complexity of Algorithm
\ref{algo:dagpre} is linear to the number of interactions, as for each
examined edge its interactions are processed at most once (from the
temporally earliest to the latest) \cite{cormen2009introduction}.
Each edge is checked for deletion
at most twice (once as an outgoing edge and at most once as an
incoming edge).
Topological sorting of the vertices (in the beginning of the
algorithm)
examines each edge of the DAG once.
Hence, the algorithm is very fast and can potentially result in
significant cost savings in maximum flow computation, as we
demonstrate in Section \ref{sec:experiments}. 
%The source node is deleted

\begin{algorithm}
\begin{algorithmic}[1]
%\LinesNumbered
%\scriptsize
\small
\Require DAG $G(V,E)$
\State define topological order for $G$'s vertices
\For{each vertex $v\in V\backslash\{s,t\}$ in topological order}
  \If{$v$ has no incoming edges}
      \State delete all outgoing edges from $v$
      \State delete $v$ from $V$
  \Else
      \State $mintime=\min_{(w,v)\in E}\{\min_{(t,q)\in (w,v)_S} t\}$
      \For{each $(v,u)\in E$}
          \For{each $(t,q)\in (v,u)_S$}
              \If{$t<mintime$}
               \State delete $(t,q)$ from $(v,u)_S$
              \EndIf
          \EndFor
          \If{$(v,u)_S=\emptyset$} \Comment{all interactions deleted}
               \State delete $(v,u)$ from $E$
          \EndIf
      \EndFor
      \If{$v$ has no outgoing edges}
         \State delete $v$ from $V$ 
         \State delete from $E$ all edges $(w,v)$ incoming to $v$ and
         \State recursively delete all $w\in V$ with no outgoing edges
      \EndIf
  \EndIf
\EndFor
\end{algorithmic}
\caption{DAG preprocessing algorithm}
\label{algo:dagpre}
\end{algorithm}

\subsubsection{Graph simplification}\label{sec:simplification}
Before applying LP, we also propose a {\em graph
  simplification} approach that can reduce the cost of maximum flow computation.
This approach is based on our observation that chains
which originate from the source vertex can
be reduced to single edges.
In a nutshell, graph simplification iteratively identifies and reduces
such chains by applying the greedy algorithm on them, until no further
reduction can be performed. The resulting graph is then solved using LP.

We start by showing that any chain that starts from the source of the
graph can be converted to a single edge without affecting the
correctness of maximum flow computation in the graph.
The interactions on the single edge that replaces the chain
are all interactions that enter the sink (i.e., the destination
vertex) of the chain and result in
increasing its buffered quantity.
For example, the entire
chain  of Figure \ref{fig:greedyopt}(a) can be
reduced to a single edge $(s,t)$ with interactions $\{(6,3),(8,4)\}$.
To derive this edge, we have to run the greedy algorithm on the graph
and define one interaction on $(s,t)$ for each interaction $(t_i,q_i)$ in $(y,t)$
that increases buffer $B_t$.
The defined interaction on is $(t_i,\min\{q_i,B_y^{t_i}\})$.
Each such interaction
corresponds to transferring a quantity from $s=source(G)$ to $t$ through
the other nodes. Hence, at any time moment, $B_t$ in $G$ is equivalent
to $B_t$ in the transformed graph.
%To derive the interactions on $G'$, we run the
%greedy algorithm on the chain $G$ and for each interaction that
%has the sink $t$ as target, we define an interaction on the
%transformed DAG which carries the quantity that is added to the buffer
%$B_t$.
In general, the following lemma holds.

\begin{lemma}\label{lem:simplification}
  Let $G$ be an interaction network and let $s$ be the source node of $G$.
  Assume a chain of $k$ vertices  $sv_1v_2\dots v_k$,
  i.e.,  for each $v_i, i<k$, the in- and out-degree of $v_i$ is 1.
  Then, $G$ can then be reduced to a graph $G'(V',E')$,
  where $V'=V-\{v_1,v_2,\dots,v_{k-1}\}$ and
  $E'=E-\{(s,v_1),(v_1,v_2),\dots,$\linebreak$(v_{k-1},v_k)\}+(s,v_k)$.
  The interactions on the new edge $(s,v_k)$ are those that
  determine the total quantity buffered at $v_k$ after running the
  greedy algorithm on chain $sv_1v_2\dots v_k$.
  Then, the maximum flow throughout $G$ is equal to the
  maximum flow throughout $G'$.
\end{lemma}

\begin{proof} [Sketch]
Recall that reserving flow in the source vertex $s$ of $G$ cannot
increase the maximum flow that reaches its sink.
The same holds for all vertices $\{v_1,v_2,\dots,v_{k-1}\}$
in a chain $sv_1v_2\dots v_k$
that originates from the
source $s$, except from the last vertex $v_k$.
Hence, by running the greedy algorithm and replacing chain\linebreak
$sv_1v_2\dots v_k$ by an edge $(s,v_k)$ having
all interactions that increase buffer $B_{v_k}$ does not affect the
correctness of maximum flow computation in $G$,
as the quantity received by $v_k$ via chain
$sv_1v_2\dots v_k$
at any time is equivalent to the
quantity received by $v_k$ via the new edge $(s,v_k)$ at any time.
\end{proof}

%The example is shown in
Figure \ref{fig:simplification}
illustrates
the effect of Lemma \ref{lem:simplification} and exemplifies our 
simplification approach.
Assume that the initial graph is shown in Figure
\ref{fig:simplification}(a).
After reducing the two chains that originate from the sink to edges,
the graph is simplified as shown in Figure
\ref{fig:simplification}(b). Note that the reduction of chain $syx$
introduces a new edge $(s,z)$ with interactions $\{(3,2),(7,1)\}$,
however, an edge $(s,z)$ already exists in the graph with interactions
 $\{(2,5),(11,2)\}$. In such a case,
the two edges are {\em merged} to a single edge with all four
interactions as shown in Figure \ref{fig:simplification}(c).
After the merging, a new chain $szw$ that originates from the source
$s$ is created. This chain is then reduced to single edge $(s,w)$ as
shown in  Figure \ref{fig:simplification}(d).
At this stage the graph cannot be simplified  any further, so we compute
its maximum flow using LP. Note that the LP optimization problem of
the initial graph in Figure \ref{fig:simplification}(a) has 9 variables
(as many as the interactions that do not originate from $s$), whereas
the reduced graph in Figure \ref{fig:simplification}(d) has only 3
variables. This demonstrates the reduction to the cost of solving the
problem achieved by our graph simplification approach.

\begin{figure}[tbh]
\centering
\small
\begin{tabular}{@{}c c @{}}
\includegraphics[width=0.48\columnwidth]{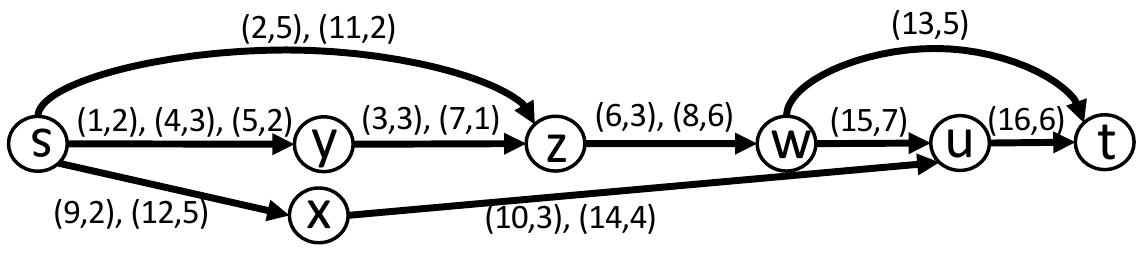}&
\includegraphics[width=0.48\columnwidth]{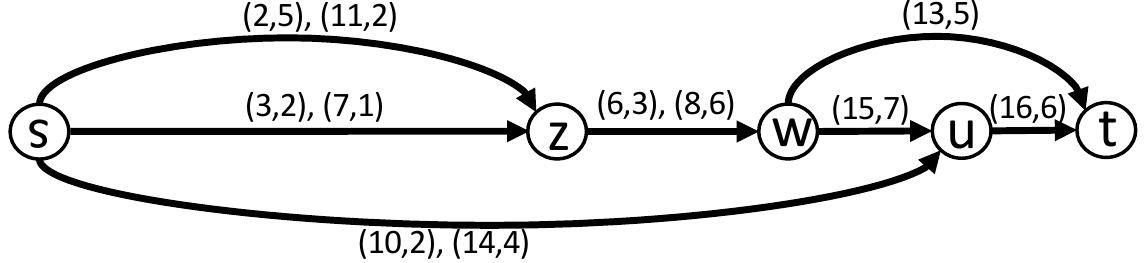}\\
(a) initial graph&
(b) first chain reduction\\
\includegraphics[width=0.48\columnwidth]{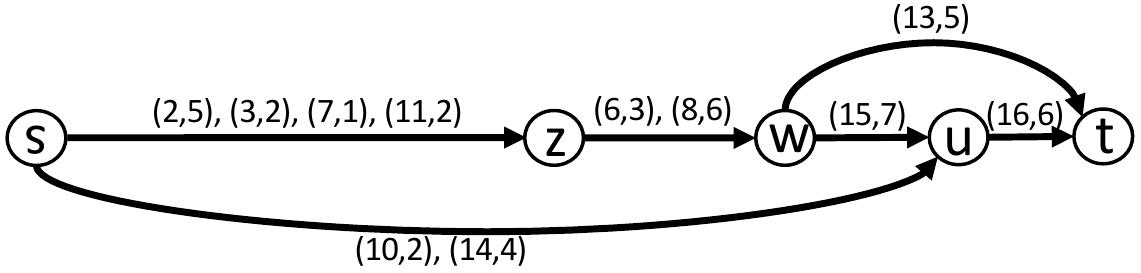}&
\includegraphics[width=0.48\columnwidth]{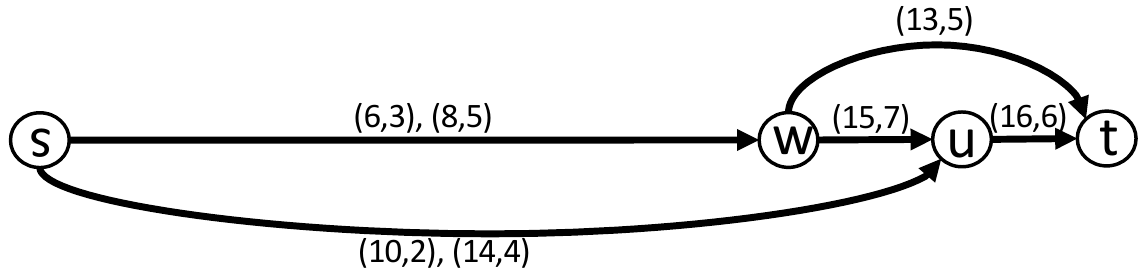}\\
(a) edge merging&
(b) second chain reduction\\
\end{tabular}
\caption{Example of graph simplification}
\label{fig:simplification}
\end{figure}

A pseudocode for the proposed graph simplification approach is Algorithm
\ref{algo:dagsimplification}.
Since each edge is examined just once before being reduced,
the complexity of the algorithm
is linear to the number of interactions on all edges that are removed
(i.e., those processed by executions of the greedy algorithm) and,
overall, linear to the number of interactions in the graph.
On the other hand, simplification can result in significant cost savings
in maximum flow computation, as already discussed and as we
demonstrate in Section \ref{sec:experiments}. 
%The source node is deleted

\begin{algorithm}
\begin{algorithmic}[1]
%\LinesNumbered
%\scriptsize
\small
\Require Graph $G(V,E)$
\While{$G$ contains a chain $C$ originating from source $s$}
    \State run Greedy to simplify chain $C$ to interaction set $I$
    \State remove edges $\{(s,v_1), (v_1,v_2),\dots (v_2,v_k)\}$ from $E$ 
    \If $(s,v_k)\notin E$ \Comment{edge $(s,v_k)$ does not exist}
      \State add edge $(s,v_k)$ to $E$
    \EndIf
    \State add $I$ to interaction set of edge $(s,v_k)$
\EndWhile
\end{algorithmic}
\caption{Graph simplification algorithm}
\label{algo:dagsimplification}
\end{algorithm}

\section{Flow Pattern Enumeration}\label{sec:algorithm}
%As an application of flow computation in DAGs, in this section, we
%will briefly study the problem of enumerating DAG patterns in large
%graphs. Specifically, given a {\em DAG pattern} $G_P$, the objective is to
%find all instances of the pattern in a large graph $G$ and compute the
%flow throughout the instance 
%Consi
%One application of flow computation in DAGs is the enumeration of
%DAG patterns that repeat themselves and the computation of their flow(s).
In the previous section, we have discussed the problem of computing
the flow throughout a subgraph of the interaction network.
We now turn our attention to flow pattern search in large graphs.
As defined in Section \ref{sec:def}, a pattern  $G_P(V_P,E_P)$ is a DAG
and its instances are subgraphs of the input graph $G(V,E)$.
To compute the flow throughout an  instance of a pattern, we can use
the algorithms presented in Section \ref{sec:flowcomp}.
In this section, we present techniques for finding the pattern instances and their flows.
%We propose two efficient algorithms for different type of patterns
%(e.g., chain patterns, with parallel edges or artificial nodes).
%The inputs of the problem are a graph $G(V,E)$  and
%a network pattern.
%The goal is compute the pattern instances and their flow.
%that will arrive at the end of process.
As discussed
in Section \ref{sec:related},
pattern matching is a well-studied problem,
but most previously proposed techniques apply on labeled graphs and
all of them disregard flow computation.
Our goal here is to demonstrate that
the enumeration of pattern instances and
their flows in an interaction network can
greatly benefit from a simple graph
preprocessing technique.
Before discussing it, we present a baseline graph browsing approach.

\subsection{A graph browsing approach}\label{sec:algo:direct}
A direct approach to solve the pattern search problem traverses the
graph, trying to identify matches of the pattern $G_P$ by expanding
{\em partial matches} of $G_P$.
As discussed in previous work \cite{DBLP:journals/pvldb/SunWWSL12},
graph browsing could be the most efficient approach, especially for
pattern search in unlabeled graphs, where the number of instances can
be numerous.
Specifically, in graph browsing, the vertices of $G_P$
are considered in a topological order. Starting from the source vertex
of $G_P$, for each vertex $v_P\in G_P$,
$v_P$ is mapped to a vertex $v\in G$,
making sure that all structural and mapping  ($\mu$) constraints
w.r.t. all previously instantiated vertices are satisfied. 
%by instantiating the incoming the edges to $v_P$, which stem from
%already instantiated vertices $u_P\in G_P$, which precede $v_P$ in the
%topological order.
For example, consider the pattern $G_P$ of Figure
\ref{fig:example}(b) and the graph $G$ of Figure \ref{fig:example}(a).
To find all matches of $G_P$ in $G$, we instantiate the first vertex
$a$ of $G_P$ to each of the four vertices of $G$ and for each instance
of $a$, we perform graph browsing to gradually ``complete'' possible
matches (using a backtracking algorithm). That is, from $a=u_1$, we
follow the outgoing edge of $u_1$ to instantiate $b=u_2$; then, the
outgoing edge of $b=u_2$ to instantiate $c=u_3$; then, the first
outgoing edge of $c=u_3$ to instantiate the sink vertex $a=u_1$, which
gives us the pattern match $u_1u_2u_3u_1$. Then, we backtrack and try
the instantiation $a=u_4$, which fails, because $u_4\neq u_1$ (recall that
$u_1$ is already mapped to the source $a$ of $G_M$).
For each pattern instance computed by this method, we can use the
approaches proposed in Section \ref{sec:flowcomp} to compute the
corresponding flow.

Note that for certain patterns, like the chain pattern of Figure
\ref{fig:example}(b), for which the maximum flow can be computed by
the greedy algorithm, we can compute the maximum flows of
their instances by gradually computing the interactions that determine
the maximum flows of their partial matches
(similarly to the simplification approach that we have
proposed in Section
\ref{sec:simplification}). For example, after finding the partial match
$u_1u_2u_3$, we apply the greedy approach to derive the set of
interactions $\{(3,\$4), (5,\$2)\}$, which determine the maximum flow
into $u_3$ originating from $u_1$ at any time moment.
After we expand to complete the match $u_1u_2u_3u_1$,
we can compute its flow {\em incrementally},
from the set of interactions $\{(3,\$4), (5,\$2)\}$ into $u_3$,
by using only this set and the interactions on edge $u_3u_1$ in the
greedy algorithm. If
$u_1u_2u_3$ was expanded to another pattern match, we could still use
the same set $\{(3,\$4), (5,\$2)\}$ to compute its flow incrementally,
without having to run the greedy
algorithm for the entire set of interactions in the new instance. 

The advantage of the graph browsing pattern enumeration approach is
that it is a general method that does not require any precomputed
information. At the same time, it is expected to be reasonably
efficient, because there is not much room for pruning vertices as not
being candidates to be mapped to pattern vertices (recall that graph
vertices are {\em unlabeled} and mapping is only based on
equality/inequality constraints to other mapped vertices).
In the next subsection, we propose a graph preprocessing
approach that facilitates faster pattern enumeration.  

\subsection{A preprocessing-based approach}\label{sec:algo:prepro}
We
assume that the graph $G$ is static (i.e., it contains historical
data).%
\footnote{For the case of graphs which grow over time,
we can apply delta-updates to the precomputed data, to consider
interactions that enter $G$ after the initial precomputation.}
We propose the preprocessing of $G$ and the extraction from it instances
of certain subgraphs that can help in identifying instances of larger
patterns that include the subgraphs.
The intuition behind this approach is that we can avoid searching for
subgraphs of the pattern from scratch; instead, we can retrieve the
pattern's structural components (and precomputed flow data) and then
``stitch'' them together using join algorithms.
This is not a new idea, as the extraction and indexing of subgraphs in
order to facilitate graph pattern search has been used in several
studies \cite{DBLP:journals/pvldb/SunWWSL12, DBLP:conf/icde/ChengYDYW08}.
Here, we employ the idea in the context of flow pattern enumeration.
%\todo{add reasoning and differences to previous work}

\stitle{Path Precomputation.}
The subgraphs we precompute are paths up to a certain length
(i.e., up to $k$ hops).
We form one table for each length, holding all paths of that length.
That is, for each path, we store:
%store all these paths in 
%For paths up to $k$ hops, we precompute their
%structural matches in the graph $G$ and we store in a table for each
%instance:
(i) the sequence of vertex-ids that form the path, (ii) the
sequence of interactions $e_S$
that enter the buffer $B_t$ of the sink $t$ of the path,
after applying the greedy algorithm;
$e_S$ determines the flow from the
source of the path to the sink at any time moment. 

\stitle{Enumeration of Pattern Instances.}
%In a nutshell, to compute the matches of $G_P$, we (i) decompose $G_P$ to
%paths, (ii) access and join any materialized tables that can help in
%the construction of the structural matches, (iii) use the precomputed flow data 
%to avoid or minimize flow computations wherever possible.
To enumerate all pattern instances using the precomputed tables,
the first step is to identify precomputed path subpatterns in $G_P$
and access and join the corresponding tables, in order to form either
complete instances of $G_P$, or partial ones  if complete instances
cannot be derived simply by combining paths from the accessed
tables.
In the latter case,
we use the graph representation to verify the existence of any missing
edges in the partial instance and/or to expand from the partial instances
and include missing vertices and edges (or determine that the partial
instance cannot be expanded to a complete one).
As soon as a complete pattern match is identified, we compute the flow
of the graph. While doing so, we use any
precomputed flows from the tables wherever possible to avoid flow computations. 

Consider, for example, the flow pattern $G_P$ shown in Figure
\ref{fig:fp_exam}(a).
Assume that we have preprocessed and have available all instances of
two-hop and three-hop cyclic paths that start from and end to the same node
$a$ in two tables  $L_2$ and $L_3$, respectively.
In this case, we can easily compute all instances of $G_P$, by only accessing
and using preprocessed data. Specifically, if the preprocessed paths
are sorted by vertex-id,%
\footnote{This is easy to achieve if the paths are
  computed by a DFS algorithm that 
  considers the graph vertices as starting vertices of the DFS in sorted order.}
it suffices to scan $L_2$ and $L_3$ and merge-join them, in order to
find all pairs of paths from $L_2$ and $L_3$ that have the same start
(and hence end) vertex. For each such pair, we verify the remaining
constraint (that $b$ and $c$ are mapped to vertices different than the one
whereto $e$ is mapped).
Finally, to compute the total flow of the resulting pattern instance,
we sum up all precomputed incoming flows to the sinks of the two paths.

\begin{figure}[h]
\centering
\small
\begin{tabular}{@{}c c @{}}
\includegraphics[width=0.28\columnwidth]{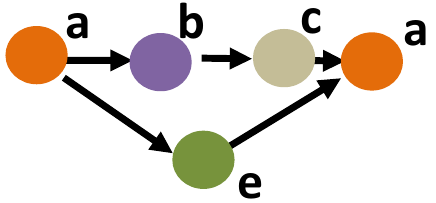}&
\includegraphics[width=0.28\columnwidth]{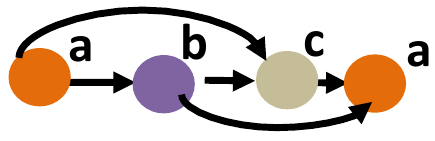}\\
(a) easy pattern&
(b) hard pattern\\
\end{tabular}
\caption{Examples of flow patterns}
\label{fig:fp_exam}
\end{figure}

On the other hand, the precomputed data may not be fully utilizable,
when computing the instances of patterns such as the one in Figure
\ref{fig:fp_exam}(b).
To enumerate the instances of such a pattern, we can
first scan the $L_3$ table and
for each accessed 3-hop cycle $\mu(a)\to \mu(b)\to \mu(c)\to \mu(a)$,
access the input graph $G(V,E)$
to verify
whether there is an edge that
connects the vertex $\mu(a)\in V$ mapped to $a$
to the vertex $\mu(c)\in V$ mapped to $c$
and
whether there is an edge that
connects the vertex $\mu(b)\in V$ mapped to $b$
to the vertex $\mu(a)\in V$ mapped to $a$.
If these edges do exist, they are retrieved and combined with the
edges
along the path
$\mu(a)\to \mu(b)\to \mu(c)\to \mu(a)$ to form an instance,
which is then passed to the algorithms of Section \ref{sec:flowcomp}
for flow computation.
In this case, the precomputed flows of paths in $L_3$ cannot be used
because the paths are not isolated in the instances of the pattern.

In general, precomputed flows along paths can be useful only for
pattern instances wherein these paths are independent and can be
progressively simplified using the technique proposed
in Section \ref{sec:simplification}.
Still, even when precomputed flows are not useful, the precomputed
paths can be used to accelerate finding the instances of the patterns.

We have assumed so far that the precomputed path tables 
%with the structural matches of chains
are sorted by starting vertex (and by prefix in general).
This helps to reduce the cost of pattern matching
by merge-joining the tables.
If the number of path instances is not extreme, tables can also be
sorted by other columns or column indices can be
used to accelerate other cases of joins. 
In addition, it is possible to use hash tables for each of
the columns and replace joins by lookups.
%For instance, when evaluating $P_{B2}$, we can first merge-join the edges
%table with the 3-hop table and for each result, confirm if $b$ is
%linked to $d$ by an index lookup at the edges table.

\subsection{Non-rigid patterns}\label{sec:relaxed}

The patterns that we have defined so far have a rigid structure which
is determined by a DAG.
In some applications, however, certain patterns with more relaxed structure
could be of interest. Consider, for example, a money-laundering
pattern where a source node $a$ is sending payments to 
recipients (which do not have a fixed number) and then these
recipients send money back to $a$. We might be interested in identifying
instances of such patterns and their corresponding flows.
Right now, we could only define a set of different patterns and
measure their flows independently, as shown in Figure
\ref{fig:nonrigid}(a).
Then, we could aggregate the flows of all instances of the different
patterns that
correspond to the same node $a$ in order to compute the total flow
from $a$ to $a$ via other nodes.

This approach has several shortcomings. First, we would have to
compute and merge the results of multiple pattern queries. Second,
there is no limit on how many patterns we should use. Third, the final
result might not be correct, as the flows of subpatterns could be included
in the flows of superpatterns (for example, an instance of the 2nd
pattern in Figure \ref{fig:nonrigid}(a) includes two instances of the
first pattern).

In order to avoid these issues, we can define a
{\em relaxed} pattern as shown in Figure   \ref{fig:nonrigid}(b),
which links $a$ to $a$ by parallel paths
via any number of intermediate nodes.
Finding the instances of this pattern and measuring their flows is
very easy using our precomputation approach, as we only have to scan
the 2-hop cycle table $L_2$ and, for each instance of $a$, we
have to aggregate the flows of the corresponding rows of the table.
We can also set constraints to the number of paths in a non-rigid
pattern. For example, we may be interested in instances of the  
pattern shown in Figure   \ref{fig:nonrigid}(b) which include at least
10 cycles.

\begin{figure}
\centering
\begin{tabular}{cc}
\includegraphics[width=.6\columnwidth]{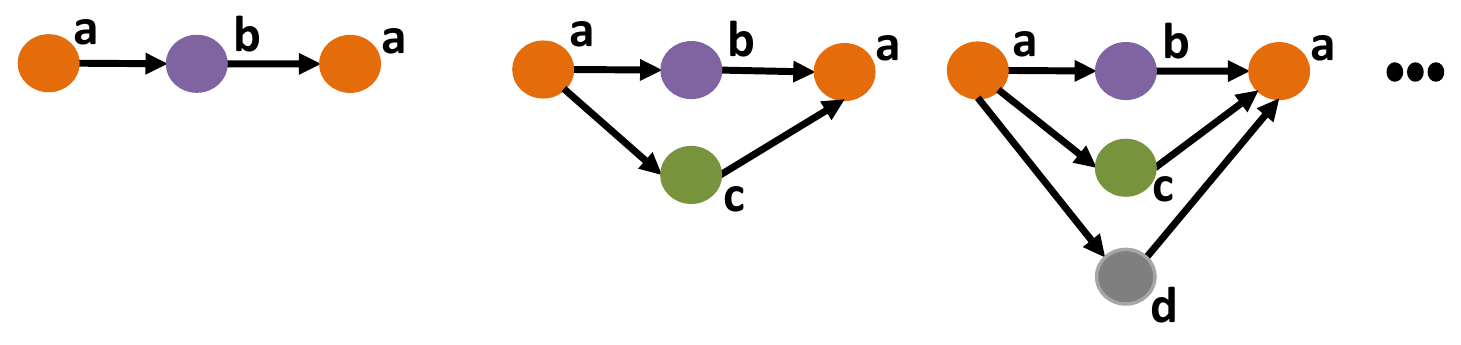}&
\includegraphics[width=.15\columnwidth]{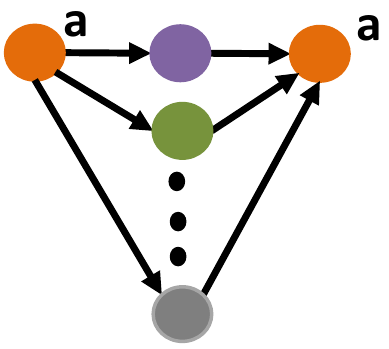}\\
(a) 2-hop rigid patterns &
(b) relaxed pattern \\
\end{tabular}
\caption{2-hop nonrigid pattern}
 \label{fig:nonrigid} \vspace{-.5em}
\end{figure}

\section{Experimental Evaluation}\label{sec:experiments}
In this section, we evaluate the performance of the flow computation
techniques proposed in Section \ref{sec:flowcomp}.
In addition, we evaluate the efficiency of the preprocessing-based
approach for flow pattern enumeration proposed in Section
\ref{sec:algorithm}.
All methods were implemented in C
and the experiments were run on a
MacBook Pro with an 2.3 GHz Quad-Core Intel Core i5 and 8GB memory.
For the implementation of LP, we used the lpsolve library%
\footnote{https://sourceforge.net/projects/lpsolve/}
(version 5.5.2.5). The source code of the paper is publicly available.%
\footnote{https://github.com/ChrysanthiKosifaki/FlowComputation} 

\subsection{Description of Datasets}\label{sec:dataset}
We used three real datasets,
generated from real interaction networks: the Bitcoin
transactions network, an internet traffic network
and a loans exchange network. We now provide details about the data.
Table \ref{table:datasets} summarizes statistics about them.

\stitle{Bitcoin:}
We downloaded all transactions in the bitcoin network \cite{Nakamoto_bitcoin:a}
up to 2013.12.28
from
http://www.vo.elte.hu/bitcoin/.
The data were collected and formatted by the
authors of \cite{DBLP:journals/corr/KondorPCV13}.
We joined tables `txedge.txt' with `txout.txt' to create a single table with
transactions of
the form (sender, recipient, timestamp, amount). We also used
table `contraction.txt' to merge addresses which belong to the
same user. Addresses were mapped to integers in a continuous range
starting from 0.
Finally, we converted all amounts to \bitcoinB~(originally in
Satoshis, where 1 Satoshi=$10^{-8}$\bitcoinB)
and removed all insignificant transactions with
amounts less than 10000 Satoshis.

\stitle{CTU-13:}
We extracted data from a botnet traffic network
\footnote{https://mcfp.felk.cvut.cz/publicDatasets/CTU-Malware-Capture-Botnet-52/},
created in CTU University\cite{garcia2014empirical}.
%information about \textbf{IP addresses}, the bytes which
%  transferred IP addresses and the time that the transaction
%  happens.
Hence, the vertices of the graph are IP addresses
and the interactions are data exchanges between them at different timestamps. 
%The dataset
%contains data exchanges between IP addresses and the times of the exchanges. 
%as vertices in our graph and the
%edges are the interactions.
We consider as flow the total amount of bytes transferred between IP addresses.

\stitle{Prosper Loans:}
Prosper%
%was invented in 2005
\footnote{https://en.wikipedia.org/wiki/Prosper\_Marketplace}
is an online peer-to-peer loan
service.
We consider Prosper as an interaction network between users
who lend money to each other.
Each record
includes the lender, the borrower, the time of the transaction
and the loan amount. We disregarded the tax that the borrower paid
for the transaction and considered only the net loan amount.
The data were downloaded from
http://konect.uni-koblenz.de.

%\makebox[1.3 \textwidth][c]{
\begin{table}[ht]
\caption{Characteristics of Datasets}
\vspace{0.3cm}
\centering
 \small
\begin{tabular}{@{}|@{~}c@{~}|@{~}c@{~}|@{~}c@{~}|@{~}c@{~}|@{~}c@{~}|@{}}
\hline
Dataset &\#nodes &\#edges & \#interactions &avg. flow \\
\hline
Bitcoin & 12M&27.7M& 45.5M & 34.4\bitcoinB\\
CTU-13 & 607K&697K &2.8M & 19.2KB\\
Prosper Loans &88K&3M&3.04M&$\$$76\\[0.2ex]
\hline
\end{tabular}
\label{table:datasets}
\end{table}
% }

\subsection{Flow Computation}\label{sec:exp:flowcomp}
In the first set of experiments, we evaluate the flow computation
techniques discussed in Section \ref{sec:flowcomp}.
For this purpose, we extracted a number of subgraphs from each network
and we applied flow computation on each of them.
Specifically, we identified seed vertices in the networks from which there
are paths (up to three hops) that pass through other vertices and then return to
the origin.
For each seed vertex, we
merged all edges along these paths to
form a single subgraph of the network.
%, where the source and the sink is the seed vertex.
Figure \ref{fig:ex_realdag} shows an example of such a subgraph,
formed by merging all paths that start from vertex 143 and end at the same vertex.

\begin{figure}[tbh]
  \centering
  \small
  \includegraphics[width=0.8\columnwidth,height=6cm]{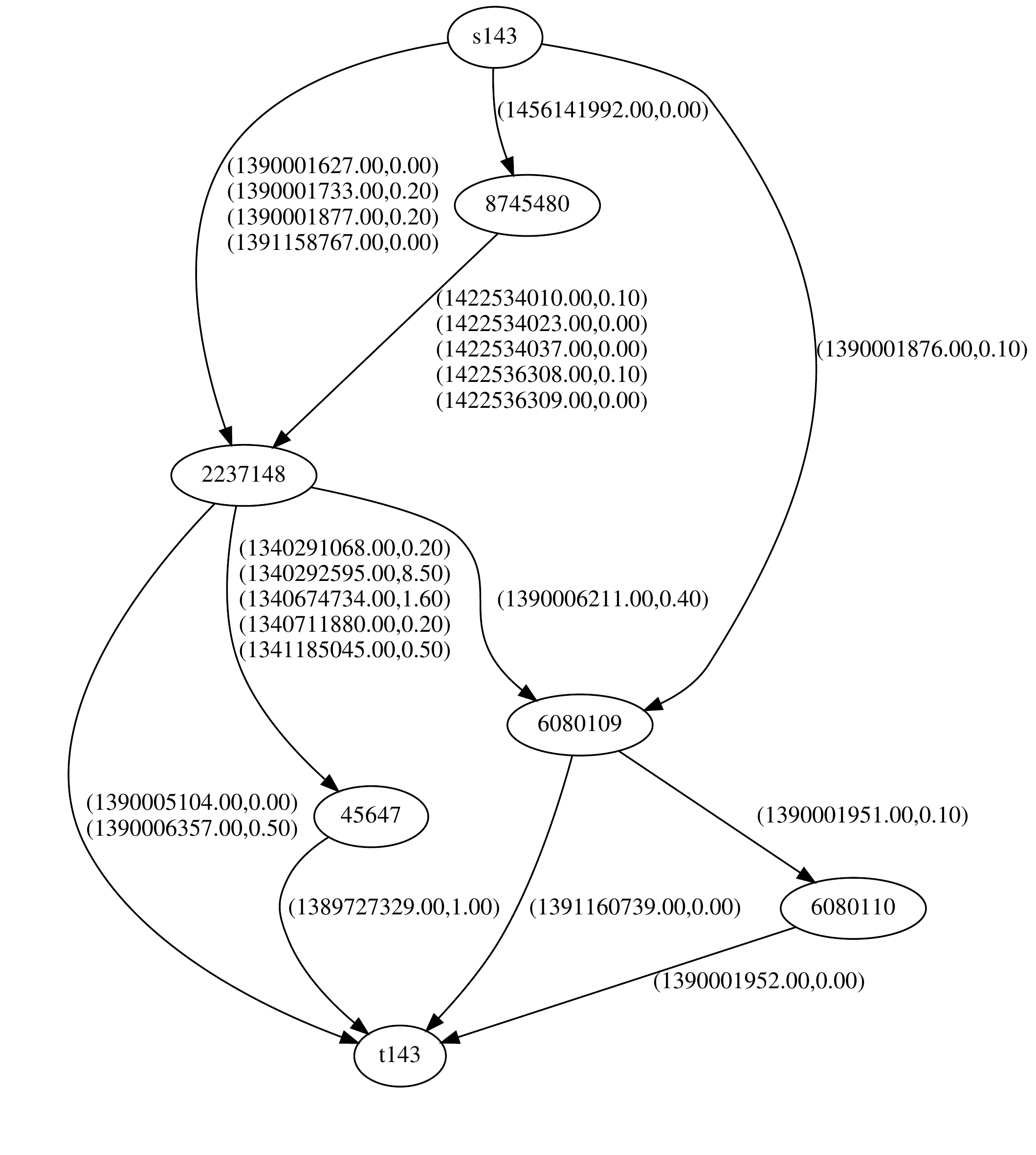}
  \caption{Example of a subgraph extracted from Bitcoin}
  \label{fig:ex_realdag}
\end{figure}

%We thwe then measured the flow, assuming that the source and the sink
%vertex are the same source.
We discarded subgraphs with more than 10K interactions because the
LP algorithm for maximum flow computation was too slow on them.
The number of tested subgraphs extracted from each
dataset and their statistics are shown in Table \ref{table:dagstats}.
The subgraphs are relatively small in terms of vertices and
edges, but they have a large average number of interactions
(compared to the expected number derived from
Table \ref{table:datasets}).
Hence, (i) these subgraphs are statistically interesting,
because they have many interactions and relatively large flow
and (ii) computing the maximum flow through them is relatively
expensive (again, due to the large number of interactions).

\begin{table}[ht]
  \small
  \caption{Statistics of subgraphs}\label{table:dagstats}
  \centering
%  \scriptsize
  \vspace{0.3cm}
  \begin{tabular}{|@{~}l@{~}|@{~}c@{~}|@{~}c@{~}|@{~}c@{~}|@{~}c@{~}|}
    \hline
 Dataset   &\#subgraphs&avg \#vertices&avg \#edges&avg \#interactions\\
    \hline
Bitcoin&48.7K&5.16&6.42&448.4\\%\cline{2-11}
CTU-13    &9235&3.24&2.49&15.9\\
Prosper Loans &137&6.1&8&611.5\\%\cline{2-11}
\hline
\end{tabular}
\end{table}

\stitle{Compared methods.}
We applied the following methods to compute the flow on the extracted
subgraphs from each dataset.
\begin{itemize}
  \item The {\bf greedy algorithm} presented in Section
    \ref{sec:flow:greedy}. This algorithm is naturally the fastest
    one, but computes the flow based on the greedy transfer
    assumption, i.e., it does not (always) find the maximum flow that
    can be transferred from the source to the sink of the graph.
  \item {\bf LP} solves the maximum flow problem using linear
    programming, as discussed in Section
    \ref{sec:maxflowcomp}, using a direct application of the LP
    solver.
  \item {\bf Pre} first applies    the
    greedy solubility test on the subgraph (explained in Section
    \ref{sec:greedyreduction}) to test whether the maximum flow can be
    computed using the greedy algorithm. In this case, it uses the
    greedy algorithm instead of LP. Otherwise, it applies the
    graph preprocessing approach (Section \ref{sec:flowcomp:prepdag})
    to remove any interactions, edges, or vertices that do not
    contribute to flow computation. If any edges and/or vertices are
    removed, it checks again for solubility by greedy. In the end, if the
    maximum flow is not guaranteed to be computed by greedy,  
    it applies LP.  
  \item {\bf PreSim} follows the steps of  Pre and if, in the
    end, LP has to be applied, PreSim attempts to further simplify the graph by
    applying the method presented in Section \ref{sec:simplification}
    which computes part of the
    maximum flow using the greedy algorithm.
    PreSim is our complete
    solution for maximum flow computation in temporal interaction networks.
\end{itemize}

\stitle{Results.}
The second rows of Tables \ref{table:expflowbtc}, \ref{table:expflowip}, and
\ref{table:expflowpl}
show the average runtime (in msec) 
of the compared flow computation methods
per tested subgraph.
%on the
%subgraphs extracted from the three real interaction networks.
The greedy algorithm is lightning fast, as its cost is
linear to the number of interactions.
%in the graphs and this number
%does not exceed 10000.
Its running time in all cases is in the order
of microseconds. 
For the maximum flow problem, the baseline LP approach is very slow
especially on the Bitcoin subgraphs, which contain the largest number
of interactions on average (see Table \ref{table:dagstats}).
With the help of the preprocessing approach (Pre), the graphs are
simplified and the cost of maximum flow computation is
reduced up to 14 times compared to LP. Note that the time for
preprocessing the graphs is included in the measured runtimes.
Finally, the graph simplification method (PreSim) further reduces the
cost at least two times compared to Pre. On average, the speedup of
our proposed maximum flow computation approach (PreSim) over
LP is 11x, 13x, and 32x on three networks. 

\begin{figure*}[t!]
\subfigure[Bitcoin Network]{
   \label{fig:exp:bitcoin_dp}
    \includegraphics[width=0.32\textwidth]{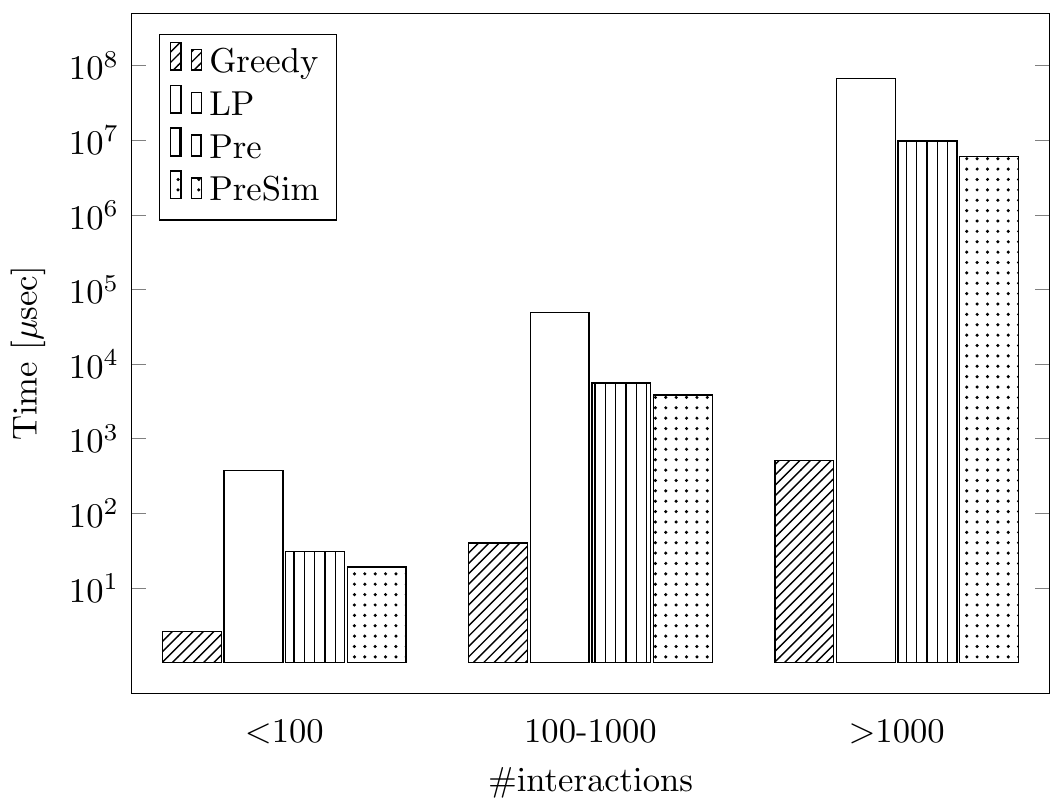}
    }\!\!\!\!
  \subfigure[Botnet Network]{
   \label{fig:exp:fb_dp}
    \includegraphics[width=0.32\textwidth]{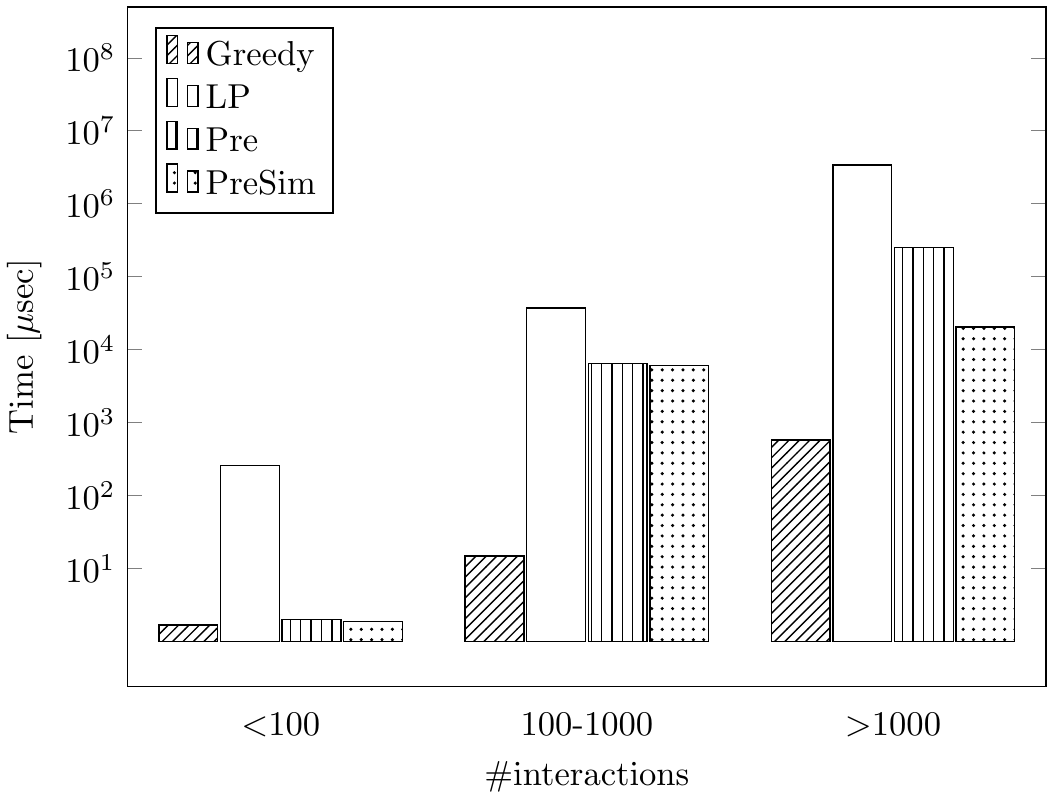}
    }\!\!\!\!
  \subfigure[Prosper Loans Network]{
   \label{fig:exp:traffic_dp}
    \includegraphics[width=0.32\textwidth]{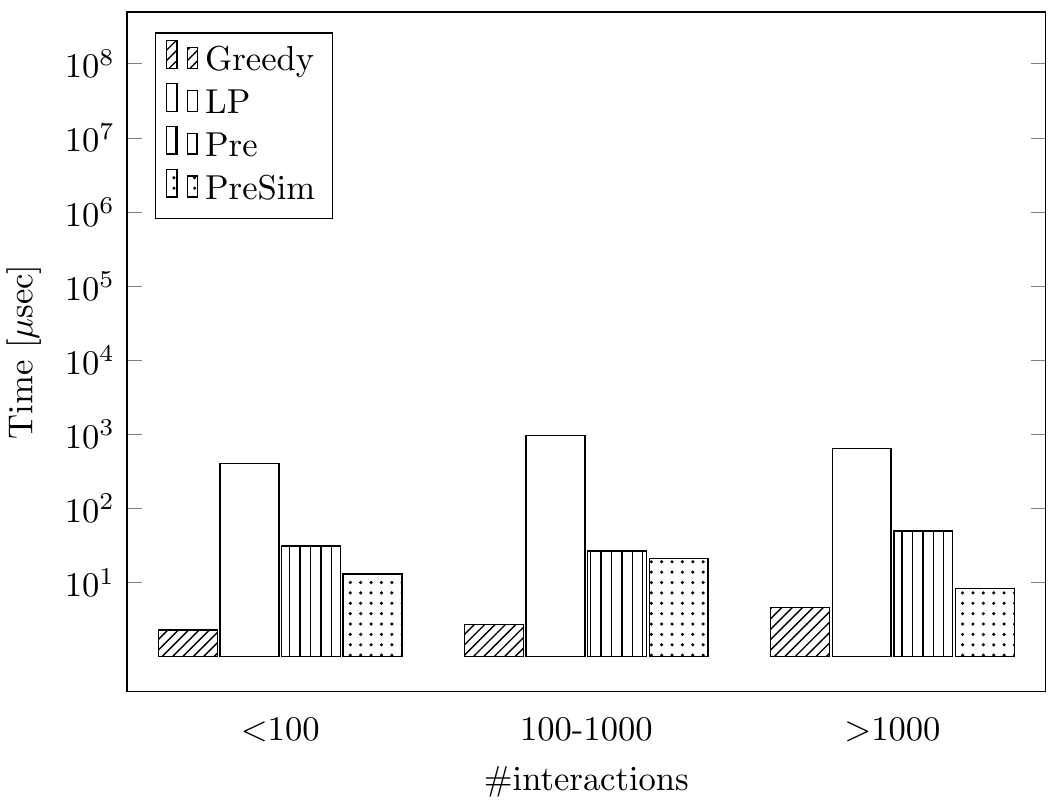}
    }
  \vspace{-0.1in}
  \caption{Runtime of algorithms as a function of the number of interactions}
  \label{fig:exp:compinter}
\end{figure*}

For a more detailed analysis of the results,
we divided the tested
subgraphs in three classes. Class A contains the easiest subgraphs,
which are found to be soluble by the greedy method. As explained
before, the cost of verifying whether a graph is soluble is very low,
so the cost of computing the maximum flow on these
graphs equals the cost of running the greedy algorithm.
Class B contains the subgraphs, which are found to be soluble by
greedy after they have been simplified by preprocessing. The cost for
computing the maximum flow on these graphs is again close to that of
the greedy algorithm. Finally, class C contains the hardest graphs,
which even after preprocessing cannot be solved using the greedy algorithm.
The last three rows of Tables
\ref{table:expflowbtc}--\ref{table:expflowpl}
average the runtimes of the tested methods on each of the three
classes of subgraphs.
%Since class A graphs are immediately solved by greedy, the
%costs of Pre and PreSim are the same as the cost of the greedy
%algorithm. The cost of these methods for class B graphs is slightly
%higher than that of greedy (since preprocessing is applied).
%Finally, the results on the hardest graphs of class C,
%show the actual improvement of PreSim over Pre (as these are the only
%cases where graph simplification is applied). 

We also divided the tested subgraphs into three categories based on the number
of interactions they include (less than 100 interactions, between 100
and 1000 interactions, more than 1000 interactions).
Figure \ref{fig:exp:compinter} compares the average performance of all methods
on each category of subgraphs at each dataset.
As expected, the costs of all methods increase with the number of
interactions. In general, the savings of PreSim and Pre over LP are
not affected by the magnitude of the problem size.
%\todo{double check results!!!}
Overall, the experiments confirm the efficiency of the
proposed techniques in Section \ref{sec:flowcomp} for greedy and
maximum flow computation.

\begin{table}[h!]
  \small
  \caption{Runtime (msec) for Bitcoin subgraphs}
  \label{table:expflowbtc}
  \centering
  \vspace{0.3cm}
  \begin{tabular}{|l|l|c|c|c|}
    \hline
    &Greedy&LP&Pre&PreSim\\
    \hline
    %\hline
    {All (48.7K)}&0.0491&5775&838.8&524.5\\
    \hline
    {Class A (35.4K)}&0.0074&2667.18&0.0078&0.0078\\
    \hline
    {Class B (7891)}&0.295&7179.39&0.575&0.575\\
    \hline
    {Class C (5366)}&0.353&24248&7615.8&4762.43\\
    \hline
  \end{tabular}
\end{table}

\begin{table}[h!]
  \small
  \caption{Runtime (msec) for CTU-13 subgraphs}
  \label{table:expflowip}
  \centering
  \vspace{0.3cm}
  \begin{tabular}{|l|l|c|c|c|}
    \hline
    &Greedy&LP&Pre&PreSim\\
    \hline
    %\hline
    {All (9235)}&0.0035&10.313&6.314&0.7902\\
    \hline
     {Class A (9199)}&0.0032&3.835&0.0033&0.0033\\
    \hline
    {Class B (3)}&0.0037&71.07&0.0074&0.0074\\
    \hline
{Class C (33)}&0.0757&1810.38&1767.5&220.2\\
    \hline
  \end{tabular}
\end{table}

\begin{table}[h!]
  \small
  \caption{Runtime (msec) for Prosper Loans subgraphs}
  \label{table:expflowpl}
  \centering
  \vspace{0.3cm}
  \begin{tabular}{|l|l|c|c|c|}
    \hline
    &Greedy&LP&Pre&PreSim\\
    \hline
    %\hline
    {All (137)}&0.0027&0.5105&0.0352&0.0157\\
    \hline
     {Class A (94)}&0.0015&0.5072&0.0016&0.0016\\
    \hline
    {Class B (25)}&0.004&0.5646&0.008&0.008\\
    \hline
    {Class C (18)}&0.0067&0.4527&0.2373&0.0889\\
    \hline
  \end{tabular}
\end{table}

\subsection{Pattern Search}\label{sec:exp:patterns}
We now evaluate the flow pattern enumeration
approaches presented in Section \ref{sec:algorithm}.
Specifically, we compare the time that
the graph browsing (GB)
approach (Section \ref{sec:algo:direct})
and the preprocessing-based (PB) approach
(Section \ref{sec:algo:prepro})
need to find the instances of several simple graph patterns and to compute
the maximum flow of each instance.
We constructed main-memory representations of the three interaction
networks that facilitate
graph browsing (i.e., we can navigate to the neighbors of each vertex
with the help of adjacency lists).

Due to the high precomputation and storage cost, from datasets Bitcoin
and CTU-13, we were able to precompute and store only paths up
to 3 hops where the start and the end vertex are the same (i.e.,
cycles). Paths of longer sizes and of arbitrary nature are multiple
times larger than the original datasets.
On the other hand, the precomputed cycles up to three hops require 
at most 20\% space
compared to the size of the entire graphs.
For the Prosper Loans dataset, we also precomputed 2-hop chains (i.e.,
paths of three different nodes) which could easily be accommodated in the main
memory of our machine.

Figure \ref{fig:testedpatterns}
shows the set of patterns that we tested in the experiments.
We experimented with six rigid patterns (P1--P6) and three 
relaxed (non-rigid) patterns (RP1--RP3). In the non-rigid patterns
(see Section \ref{sec:relaxed}), all vertices
in the parallel paths (except for the source and the sink) are
required to be different.

\begin{figure}[tbh]
  \centering
  \small
  \includegraphics[width=0.99\columnwidth]{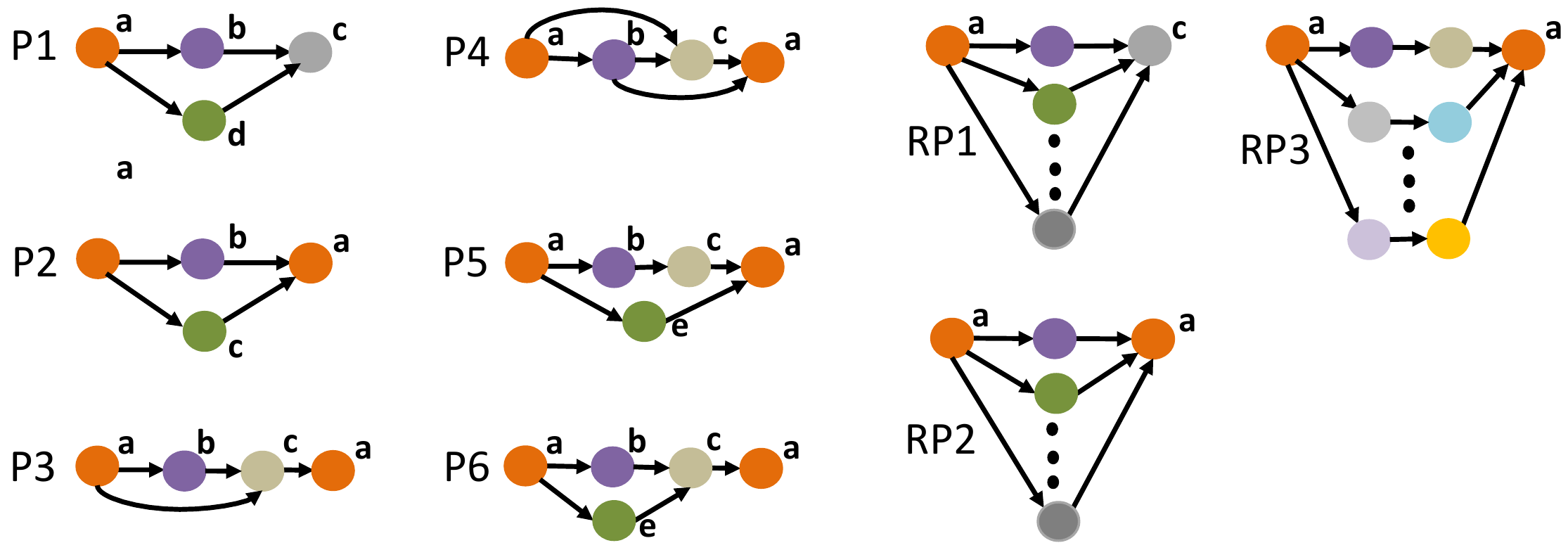}
  \caption{Set of tested patterns}
  \label{fig:testedpatterns}
\end{figure}

Tables \ref{table:patternsBTC}, \ref{table:patternsIP}, and
\ref{table:patternsPRO} compare the performance of GB to that of PB on
enumerating the instances of the various patterns and computing their
maximum flow. Note that for
Bitcoin and CTU-13 datasets, the processing times for P1 and RP1
were not included because PB was not applicable in this case (we have
not precomputed any path that would be useful).
In general, we observe that preprocessing pays off in most of the
tested patterns, as the runtimes of PB in most cases are at least one
order of magnitude lower than the corresponding ones of GB.

For some patterns and networks, prepcocessing (PB) does not give much
benefit compared to GB.
For example, for pattern P4 on the Bitcoin network
(marked with *), search by both GB and PB was terminated after finding the
first 3000 instances, because both methods are extremely slow.
The preprocessed flows cannot be used
and the maximum flow of the instances must be computed by LP.
Hence, on the Bitcoin network, PB has a similar cost as GB, as the
instances contain numerous interactions and maximum flow computation
dominates the overall cost of pattern enumeration.
The same holds for P6 (again on Bitcoin), which was terminated earlier
because both methods were quite slow.

%graph in practice
%because (i) the preprocessed flows cannot be used, (ii) pattern
%instances are not formed by joining segments of precomputed paths, but
%by retrieving from the graph structure
%for every precomputed $a\to b\to c\to a$ path the
%existence of edges $a\to c$ and $b\to a$.
%In addition, maximum flow computation on the instances of P4 was expensive
%because of the use of LP and
%of the numerous interactions and dominated the overall cost of
%pattern enumeration.

%\makebox[1.3 \textwidth][c]{
\begin{table}[ht]
\caption{Pattern Search on Bitcoin}
\vspace{0.3cm}
\centering
 \small
\begin{tabular}{@{}|@{~}c@{~}|@{~}c@{~}|@{~}c@{~}|@{~}c@{~}|@{~}c@{~}|@{}}
\hline
\textbf{Pattern} &\textbf{Instances} &\textbf{Average flow} & \textbf{GB} &\textbf{PB} \\

\hline
\textbf{P2} & 22.3G & 56.15 & 23.2 hours &30.59 sec\\
\hline
\textbf{P3} & 2.8M &4786.18  & 3155.96 sec & 179.70 sec\\
\hline
\textbf{P4*} & 3000 & 697.04 & 446.73 sec & 421.85 sec\\
\hline
\textbf{P5} &577.5M &8069.2 & 15 days (est.) & 179.74 sec\\
 \hline 
\textbf{P6*} & 2.04T & 2.81 & 1445 sec &  1059 sec\\
\hline
\textbf{RP2} & 655K & 39.86 & 422.79 sec & 53.273 msec \\
\hline
\textbf{RP3} &1.2M & 1.86&306 min& 13.53 msec\\[0.2ex]  
\hline
\end{tabular}
\label{table:patternsBTC}
\end{table}
% }

%\makebox[1.3 \textwidth][c]{
\begin{table}[ht]
\caption{Pattern Search on CTU-13}
\vspace{0.3cm}
\centering
 \small
\begin{tabular}{@{}|@{~}c@{~}|@{~}c@{~}|@{~}c@{~}|@{~}c@{~}|@{~}c@{~}|@{}}
\hline
\textbf{Pattern} &\textbf{Instances} &\textbf{Average flow} & \textbf{GB} &\textbf{PB} \\

\hline
\textbf{P2} &709M  & 2888.90 & 1952.61 sec & 762.65 msec\\
\hline
\textbf{P3} & 182 & 528.5K & 55.71 sec & 8.61 msec  \\
\hline
\textbf{P4} &91 & 1.56M &  58.564 sec& 2.518 sec \\
\hline
\textbf{P5} &208K &13116.5 &443.97 sec &4.73 msec \\ 
 \hline 
\textbf{P6} & 586 & 52892 & 410.4 sec& 14.87 msec\\
\hline
\textbf{RP2} &51266  & 11942.65 & 24.15 sec & 0.63  msec\\
\hline
\textbf{RP3} &91 &61485.58 & 375.39 sec & 0.035 msec  \\[0.2ex]  
\hline
\end{tabular}
\label{table:patternsIP}
\end{table}
% }

%\makebox[1.3 \textwidth][c]{
\begin{table}[ht]
\caption{Pattern Search on  Prosper Loans}
\vspace{0.3cm}
\centering
 \small
\begin{tabular}{@{}|@{~}c@{~}|@{~}c@{~}|@{~}c@{~}|@{~}c@{~}|@{~}c@{~}|@{}}
\hline
\textbf{Pattern} &\textbf{Instances} &\textbf{Average flow} & \textbf{GB} &\textbf{PB} \\

\hline
 \textbf{P1} & 5.12M& 45.89& 119.08 sec & 2.80 sec \\
 \hline
\textbf{P2} & 201 & 223.23 & 88.66 msec&0.004 msec\\
\hline
\textbf{P3} & 268 & 100.44 & 3.57 sec&1.3 msec \\
 \hline
\textbf{P4} &98 & 299.55& 3.54 sec& 0.723 msec \\
\hline 
\textbf{P5} &1833  & 121.47&605.67 msec & 0.021 msec \\
\hline
\textbf{P6} & 1296  &43.55  & 474.61 msec  & 11.13 msec \\
 \hline 
\textbf{RP1} & 25.5M & 25.12 & 133.37 sec & 3.01 sec  \\
\hline
\textbf{RP2} &260  & 58.061 & 0.016 msec  & 0.004 msec  \\
\hline
\textbf{RP3} &532 &10.94 & 503.89 msec & 0.040 msec\\[0.2ex]  
\hline
\end{tabular}
\label{table:patternsPRO}
\end{table}
% }

\section{Conclusions}\label{sec:conclusion}
In this paper we defined and studied the problem of flow computation
in interaction networks. We defined two models for flow computation,
one based on greedy flow transfer between vertices and one that
assumes arbitrary flow transfer and the objective is to compute the
maximum flow.
We showed that flow computation based on the first model
can be conducted very efficiently, whereas
the more interesting maximum flow computation is more
expensive. In view of this, we proposed and evaluated
a number of techniques
to reduce the cost of maximum flow computation
by at least one order of magnitude.
Note that our techniques are readily applicable for the {\em
  time-restricted} version of the problem, where we are only
interested in interactions that happen within a time window (i.e., by
simply disregarding all interactions that happened outside the window).
Finally, we studied the problem of pattern enumeration in large
graphs, where for each pattern instance, we also have to compute the
maximum flow. For this problem we proposed a technique that
precomputes the instances of simple subgraphs and their flows and uses
them to accelerate the finding of more complex patterns that have
these subgraphs as components.

Directions for future work include (i) the investigation of additional
techniques for reducing the cost of the maximum flow problem, (ii) the
investigation of similar simplification techniques to other flow
computation problems, and (iii) the automatic identification of interesting
patterns and subgraphs that have significantly more flow than expected.

\newpage

\bibliographystyle{abbrv}
\bibliography{references}

\end{document}